\documentclass[journal,romanappendices]{IEEEtran}
\IEEEoverridecommandlockouts

\usepackage{amsmath,amssymb,amsthm}
\usepackage{cite}
\usepackage{enumerate}
\usepackage{stfloats}
\usepackage{bbm}
\usepackage{hyperref}
\usepackage{flushend}

\usepackage{tikz}
\usetikzlibrary{arrows,shapes,chains,matrix,positioning,scopes}
\usepackage{pgfplots}
\usepgflibrary{shapes}
\pgfplotsset{compat=1.4}

\newtheoremstyle{custom}
{} 
{} 
{} 
{} 
{\bfseries} 
{:} 
{.25em} 
{} 
\theoremstyle{custom}

\newtheorem{theorem}{Theorem}
\newtheorem{lemma}[theorem]{Lemma}
\newtheorem{proposition}[theorem]{Proposition}
\newtheorem{definition}[theorem]{Definition}

\newtheorem{remark}[theorem]{Remark}
\newtheorem{corollary}[theorem]{Corollary}
\newtheorem{conjecture}[theorem]{Conjecture}

\newtheorem*{theorem*}{Theorem}
\newtheorem*{lemma*}{Lemma}
\newtheorem*{proposition*}{Proposition}
\newtheorem*{definition*}{Definition}
\newtheorem*{example*}{Example}
\newtheorem*{remark*}{Remark}
\newtheorem*{corollary*}{Corollary}

\renewcommand{\epsilon}{\varepsilon}

\newcommand{\mc}[1]{\mathcal{#1}}
\newcommand{\mr}[1]{\mathrm{#1}}
\newcommand{\mbb}[1]{\mathbb{#1}}

\newcommand{\indicator}[1]{\mathbbm{1}_{\left\{ {#1} \right\} }}

\newcommand{\diff}[1]{d#1}
\newcommand{\deri}[2]{\frac{d#1}{d#2}}
\newcommand{\pderi}[2]{\frac{\partial #1}{\partial #2}}
\newcommand{\dmin}{d_{\mathrm{min}}}

\newcommand{\ent}[1]{H\!\left(#1\right)}
\newcommand{\ul}[1]{\underline{#1}}

\newlength\tikzwidth
\newlength\tikzheight

\begin{document}

\title{Reed-Muller Codes Achieve Capacity \\ on Erasure Channels}
\author{Santhosh Kumar and Henry D. Pfister
\thanks{Santhosh Kumar is currently a graduate researcher in the Department of Electrical and Computer Engineering, Texas A\&M University, College Station (email: santhosh.kumar@tamu.edu).
}
\thanks{Henry~D.~Pfister is currently an Associate Professor in the Department of Electrical and Computer Engineering, Duke University (email: henry.pfister@duke.edu).
}
\thanks{This material is based upon work supported in part by the National Science Foundation (NSF) under Grant No. 1218398.
Any opinions, findings, recommendations, and conclusions expressed in this material are those of the authors and do not necessarily reflect the views of these sponsors.
Part of this work was also done while the 2nd author was visiting the Simons Institute for the Theory of Computing, UC Berkeley.
}
}
\maketitle

\begin{abstract}
This paper introduces a new approach to proving that a sequence of deterministic linear codes achieves capacity on an erasure channel under maximum a posteriori decoding.
Rather than relying on the precise structure of the codes, this method requires only that the codes are highly symmetric.
In particular, the technique applies to any sequence of linear codes where the blocklengths are strictly increasing, the code rates converge to a number between 0 and 1, and the permutation group of each code is doubly transitive.
This also provides a rare example in information theory where symmetry alone implies near-optimal performance.

An important consequence of this result is that a sequence of Reed-Muller codes with increasing blocklength achieves capacity if its code rate converges to a number between 0 and 1. 
This possibility has been suggested previously in the literature but it has only been proven for cases where the limiting code rate is 0 or 1.
Moreover, these results extend naturally to affine-invariant codes and, thus, to all extended primitive narrow-sense BCH codes.
The primary tools used in the proof are the sharp threshold property for monotone boolean functions and the area theorem for extrinsic information transfer functions.
\end{abstract}

\begin{IEEEkeywords}
affine-invariant codes,
BCH codes,
capacity-achieving codes,
erasure channels,
EXIT functions,
linear codes,
MAP decoding,
monotone boolean functions,
Reed-Muller codes.
\end{IEEEkeywords}

\section{Introduction}
\label{section:introduction}

Since the introduction of channel capacity by Shannon in his seminal paper~\cite{Shannon-bell48}, theorists have been fascinated by the idea of constructing \emph{structured} codes that achieve capacity.
The advent of Turbo codes~\cite{Berrou-icc93} and low-density parity-check (LDPC) codes~\cite{Gallager-1963,Spielman-it96,Mackay-it99} has made it possible to construct codes with low-complexity encoding and decoding that also achieve good performance near the Shannon limit.
It was even proven that sequences of irregular LDPC codes can achieve capacity on the binary erasure channel (BEC) using low-complexity message-passing decoding~\cite{Luby-it01}.
For an arbitrary binary symmetric memoryless (BMS) channel, however, polar codes~\cite{Arikan-it09} were the first provably capacity-achieving codes with low-complexity encoding and decoding.
More recently, spatially-coupled LDPC codes were also shown to achieve capacity universally over all BMS channels under low-complexity message-passing decoding~\cite{Kudekar-it11,Lentmaier-it10,Kudekar-it13,Kumar-it14}.

This article considers the performance of sequences of binary linear codes transmitted over the BEC under maximum-a-posteriori (MAP) decoding.
In particular, our primary technical result is the following.
\begin{theorem*}
 A sequence of linear codes achieves capacity on a memoryless erasure channel under MAP decoding if its blocklengths are strictly increasing, its code rates converge to some $r\in(0,1)$, and the permutation group\footnote{The permutation group of a linear code is the set of permutations on code bits under which the code is invariant.} of each code is doubly transitive.	
\end{theorem*}
Our analysis focuses primarily for the bit erasure rate under bit-MAP decoding but can be extended to the block erasure rate in some cases.
One important consequence of this is that binary Reed-Muller codes achieve capacity on the BEC under block-MAP decoding.

The main result extends naturally to $\mathbb{F}_q$-linear codes transmitted over a $q$-ary erasure channel under symbol-MAP decoding.
With this extension, one can show that sequences of Generalized Reed-Muller codes~\cite{Delsarte-ic70,Kasami-it68} over $\mathbb{F}_q$ also achieve capacity under block-MAP decoding.
For the class of affine-invariant $\mathbb{F}_q$-linear codes, which are precisely the codes whose permutation groups include a subgroup isomorphic to the affine linear group~\cite{Kasami-ic68}, one finds that these codes achieve capacity under symbol-MAP decoding.
This follows from the fact that the affine linear group is doubly transitive.
As it happens, this class also includes all extended primitive narrow-sense Bose-Chaudhuri-Hocquengham (BCH) codes~\cite{Kasami-ic68}.
Additionally, we show that sequences of extended primitive narrow-sense BCH codes over $\mathbb{F}_q$ achieve capacity under block-MAP decoding.
To keep the presentation simple, we present proofs for the binary case and discuss the generalization to $\mathbb{F}_q$ in Section~\ref{section:Fqlinear}.

These results are rather surprising.
Until the discovery of polar codes, it was commonly believed that codes with a simple deterministic structure might be unable to achieve capacity~\cite{Ahlswede-it82,Coffrey-it90}. 
While polar codes might be considered a counterexample to this statement, they require a somewhat complicated design process that is heavily dependent on the channel.
As such, their ability to achieve capacity appears somewhat unrelated to the inherent symmetry in the binary Hadamard transform.
In contrast, the performance guarantees obtained here are a consequence only of linearity and the structure induced by the symmetry of the doubly-transitive permutation group.

Reed-Muller codes were introduced by Muller in~\cite{Muller-ire54} and, soon after, Reed proposed a majority logic decoder in~\cite{Reed-ire54}.
A binary Reed-Muller code, parameterized by non-negative integers $m$ and $v$, is a linear code of length $2^m$ and dimension $\binom{m}{0}+\cdots+\binom{m}{v}$.
It is well known that the minimum distance of this code is $2^{m-v}$~\cite{Macwilliams-1977,Lin-2004,Delsarte-ic70}.
Thus, it is impossible to simultaneously have a non-vanishing rate and a minimum distance that scales linearly with blocklength.
As such, these codes cannot correct \emph{all} patterns with a constant fraction of erasures.
Until now, it was not clear whether or not these codes can correct \emph{almost all} erasure patterns up to the capacity limit.

The possibility that Reed-Muller codes might achieve capacity under MAP decoding is discussed in several works~\cite{Costello-proc07,Didier-it06,Mondelli-tcom14,Arikan-it09,Arikan-itw10,Abbe-arxiv14}.
In particular, it has been conjectured in~\cite{Costello-proc07} and observed via simulations for small block lengths in~\cite{Didier-it06,Mondelli-tcom14}.
Moreover, these codes are believed to behave like random codes \cite{Carlet-isit05,Didier-it06}.
For rates approaching either $0$ or $1$ with sufficient speed, it has recently been shown that Reed-Muller codes can correct almost all erasure patterns up to the capacity limit\footnote{It requires some effort to define precisely what capacity limit is for rates approaching $0$ or $1$. See \cite[Definition 2.5]{Abbe-arxiv14} for details.} \cite{Abbe-arxiv14}.
Unfortunately, this approach currently falls short if the rates are bounded away from 0 and 1.

Even after 50 years of their discovery, Reed-Muller codes remain an active area of research in theoretical computer science and coding theory.
The early work in~\cite{Sloane-it70,Kasami-it70,Kasami-ic76} culminated in obtaining asymptotically tight bounds (fixed order $v$ and asymptotic $m$) for their weight distribution~\cite{Kaufman-it12}.
Also, there is considerable interest in constructing low-complexity decoding algorithms~\cite{Sidelnikov-ppi92,Dumer-it04,Dumer-it06*2,Dumer-it06,Saptharishi-arxiv15}.
Undoubtedly, interest in the coding theory community for these codes was rekindled by the tremendous success of polar codes and their close connection to Reed-Muller codes~\cite{Arikan-it09,Arikan-comlett08,Mondelli-tcom14}.

Due to their desirable structure, constructions based on these codes are used extensively in cryptography~\cite{Camion-acrypt92,Ta-shma-focs01,Shaltiel-focs01,Canteaut-it01,Carlet-isit05,Carlet-it06,Didier-fse06,Gerard-alc11}.
Reed-Muller codes are also known for their \emph{locality} \cite{Yekhanin-fnt11}.
Some of the earliest known constructions for locally correctable codes are based on these codes \cite{Gemmell-stoc91,Gemmell-ipl92}.
Interestingly, local correctability of Reed-Muller codes is also a consequence of its permutation group being doubly transitive \cite{Kaufman-approxrandom10}, a crucial requirement in our approach.
However, a doubly transitive permutation group is not sufficient for local testability \cite{Grigorescu-ccc08}.

The central object in our analysis is the extrinsic information transfer (EXIT) function.
EXIT charts were introduced by ten Brink~\cite{tenBrink-elet99} in the context of turbo decoding as a visual tool to understand iterative decoding.
This work led to the area theorem for EXIT functions in~\cite{Ashikhmin-it04} which was further developed in~\cite{Measson-it08}.
For a given input bit, the EXIT function is defined to be the conditional entropy of the input bit given the outputs associated with all~\emph{other} input bits.
The average EXIT function is formed by averaging all of the bit EXIT functions.
We note that these functions are also instrumental in the design and analysis of LDPC codes~\cite{RU-2008}.

An important property of EXIT functions is the EXIT area theorem, which says that the area under the average EXIT function equals the rate of the code.
The value of the EXIT function at a particular erasure value is also directly related to the bit erasure probability under bit-MAP decoding.
For a sequence of binary linear codes with rate $r$ to be capacity achieving, the bit erasure probability, and therefore the average EXIT function, must converge to 0 for any erasure rate below $1-r$.
Since the areas under the average EXIT curves are fixed to $r$, the EXIT functions for these codes must also converge to 1 for any erasure rate above $1-r$.
Thus, the EXIT curves must exhibit a \emph{sharp transition} from $0$ to $1$, and as a consequence of area theorem, this transition must occur at the erasure value of $1-r$.

We investigate the threshold behavior of EXIT functions for certain binary linear codes via sharp thresholds for monotone boolean functions~\cite{Boucheron-2013,Kalai-05}.
The general method was pioneered by Margulis~\cite{Margulis-ppi74} and Russo~\cite{Russo-prf82}.
Later, it was significantly generalized in~\cite{Talagrand-gafa93} and~\cite{Friedgut-procams96}.
This approach has been applied to many problems in theoretical computer science with remarkable success \cite{Friedgut-procams96,Friedgut-jams99,Dinur-am05}.
In the context of coding theory, Z\'emor introduced this approach in \cite{Zemor-94}.
It was refined further in~\cite{Tillich-cpc00}, and also extended to AWGN channels in~\cite{Tillich-it04}.
For the BEC, \cite{Zemor-94,Tillich-cpc00} show that the block erasure rate jumps sharply from $0$ to $1$ as the minimum distance of the code grows.
However, this approach does not generalize directly to EXIT functions and, therefore, does not establish the location of the threshold.
To show the threshold behavior for EXIT functions, we instead focus on symmetry~\cite{Friedgut-procams96} and require that the codes of interest have permutation groups that are doubly transitive.

After we completed the first version of this paper~\cite{Kumar-arxiv15v1}, we discovered that the same approach was being pursued independently by Kudekar, Mondelli, \c{S}a\c{s}o\u{g}lu, and Urbanke~\cite{Kudekar-arxiv15v1}.

The article is organized as follows.
Section \ref{section:preliminaries} includes necessary background on EXIT functions, permutation groups of linear codes, and capacity-achieving codes.
Section \ref{section:sharp_threshold} deals with the threshold behavior of monotone boolean functions.
In Section \ref{section:applications}, as an application of the hitherto analysis, we show that Reed-Muller and extended primitive narrow-sense BCH codes achieve capacity.
Finally, we provide extensions, open problems in Section \ref{section:discussion}, and concluding remarks in Section \ref{section:conclusion}.


\section{Preliminaries}
\label{section:preliminaries}
This article deals primarily with binary linear codes transmitted over erasure channels and bit-MAP decoding.
In the following, all codes are understood to be proper binary linear codes with minimum distance at least $2$, unless mentioned otherwise.
Recall that a linear code is proper if no codeword position is $0$ in all codewords.
Let $\mc{C}$ denote an $(N,K)$ binary linear code with length $N$ and dimension $K$.
The rate of this code is given by $r \triangleq K/N$.
Denote the minimum distance of $\mc{C}$ by $\dmin$.
We assume that a random codeword is chosen uniformly from this code and transmitted over a memoryless BEC.
In the following subsections, we review several important definitions and properties related to this setup.

Notational convention: 
\begin{itemize}
\item The natural numbers are denoted by $\mathbb{N}=\{1,2,\ldots\}$.

\item For $n \in \mathbb{N}$, let $[n]$ denote the set $\{1,2,\dots,n\}$.

\item We associate a binary sequence in $\{0,1\}^N$ with a subset of $[N]$ defined by the non-zero indices in the sequence.
  We use this equivalence between sets and binary sequences extensively.
  For example, a sequence $1001100$ is identified by the subset $\{1,4,5\} \subseteq [7]$ and vice versa.
  Similarly, if $101110$ is a codeword in $\mc{C}$, then we say $\{1,3,4,5\} \in \mc{C}$.

\item We say that a set $A$ \emph{covers} set $B$ if $B \subseteq A$.
  Also, for sequences $\ul{a},\ul{b} \in\{0,1\}^N$, we write $\ul{a} \leq \ul{b}$ if $a_i \leq b_i$ for $i \in [N]$.
  Equivalently, $\ul{a} \leq \ul{b}$ if the set associated with $\ul{b}$ covers the set associated with $\ul{a}$.

\item For a set $A$, $\mathbbm{1}_{A}(\cdotp)$ denotes its indicator function.
  The random variable $\indicator{\cdot}$ is an indicator of some event. 
  For example, for random variables $X$ and $Y$, $\indicator{X \neq Y}$ indicates the event $X \neq Y$.

\item For a vector $\ul{a}=(a_1,a_2,\dots,a_N)$, the shorthand $\ul{a}_{\sim i}$ denotes $(a_1,\dots,a_{i-1},a_{i+1},\dots,a_N)$.

\item $0^n$ and $1^n$ denote the all-zero and all-one sequences of length $n$, respectively.

\item A memoryless BEC with erasure probability $p$ is denoted by $\mr{BEC}(p)$.  If the erasure probability is different for each bit, then we write $\mr{BEC}(\ul{p})$, where $\ul{p}=(p_1,\ldots,p_n)$ and $p_i$ indicates the erasure probability of bit $i$.

\item For a quantity $\theta$ with index $n$, we use either $\theta_n$ or $\theta^{(n)}$.
  Typically, we write $\theta^{(n)}$ when using $\theta_n$ may cause confusion with another quantity such as $\theta_i$; in the latter case we write $\theta_i^{(n)}$.

\item For a permutation $\pi \colon [N] \to [N]$ and $A \subseteq [N]$, $\pi(A)$ denotes the set $\{\pi(\ell) | \ell \in A\}$.
  For sequence $\ul{a} \in \{0,1\}^N$, $\pi(\ul{a})$ denotes the length-$N$ sequence where the $\pi(i)$-th element is $a_{i}$.

\item As is standard in information theory, $\ent{\cdot}$ denotes the entropy of a discrete random variable and $\ent{\cdot | \cdot}$ denotes the conditional entropy of a discrete random variable in bits.

\item All logarithms in this article are natural unless the base is explicitly mentioned.

\end{itemize}

\subsection{Bit and Block Erasure Probability}
\label{subsection:errorprobability}
The input and output alphabets of the BEC are denoted by $\mathcal{X}=\{0,1\}$ and $\mathcal{Y} = \{0,1,*\}$, respectively. 
Let $\ul{X}=(X_1,\dots,X_N) \in \mathcal{X}^N$ be a uniform random codeword and  $\ul{Y}= (Y_1,\dots,Y_N)\in \mathcal{Y}^N$ be the received sequence obtained by transmitting $\ul{X}$ through a $\mr{BEC}(p)$.
In this article, our main interest is the bit-MAP decoder.
But, we will also obtain some results for the block-MAP decoder indirectly based on our analysis of the bit-MAP decoder.

For linear codes and erasure channels, it is possible to recover the transmitted codeword if and only if the erasure pattern does not cover any codeword.
To see this, fix an erasure pattern and observe that adding a codeword to the input sequence causes the output sequence to change if and only if the erasure pattern does not cover the codeword.
Similarly, it is possible to recover bit $i$ if and only if the erasure pattern does not cover any codeword where bit $i$ is non-zero.
Whenever bit $i$ cannot be recovered uniquely, the symmetry of a linear code implies that set of codewords matching the unerased observations has an equal number of $0$'s and $1$'s in bit position $i$.
In this case, the posterior marginal of bit $i$ given the observations contains no information about bit $i$.

Let $D_i \colon \mathcal{Y}^N \to \mathcal{X} \cup \{*\}$ denote the bit-MAP decoder for bit $i$ of $\mc{C}$.
For a received sequence $\ul{Y}$, if $X_i$ can be recovered uniquely, then $D_i(\ul{Y})=X_i$.
Otherwise, $D_i$ declares an erasure and returns $*$.
Let the erasure probability for bit $i \in [N]$ be 
\begin{align*}
  P_{b,i} \triangleq \Pr(D_i(\ul{Y}) \neq X_i) ,
\end{align*}
and the average bit erasure probability be
\begin{align*}
  P_b \triangleq \frac{1}{N} \sum_{i=1}^N P_{b,i} .
\end{align*}

Whenever bit $i$ can be recovered from a received sequence $\ul{Y}=\ul{y}$, $H(X_i | \ul{Y}=\ul{y})=0$.
Otherwise, the uniform codeword assumption implies that the posterior marginal of bit $i$ given the observations is $\Pr (X_i=x|\ul{Y}=\ul{y})=\frac{1}{2}$ and $H(X_i | \ul{Y}=\ul{y})=1$.
This immediately implies that
\begin{align*}
  P_{b,i} &= \ent{X_i \mid \ul{Y}}, & P_b &= \frac{1}{N} \sum_{i=1}^N \ent{X_i \mid \ul{Y}} .
\end{align*}

Let $D \colon \mathcal{Y}^N \to \mathcal{X}^N \cup \{*\}$ denote the block-MAP decoder for $\mc{C}$.
Given a received sequence $\ul{Y}$, the vector $D(\ul{Y})$ is equal to $\ul{X}$ whenever it is possible to uniquely recover $\ul{X}$ from $\ul{Y}$.
Otherwise, $D$ declares an erasure and returns $*$.
Therefore, the block erasure probability is given by
\begin{align*}
  P_B \triangleq \Pr(D(\ul{Y}) \neq \ul{X}) .
\end{align*}

Using the set equivalence
\begin{align*}
  \{ D(\ul{Y}) \neq \ul{X} \} = \bigcup_{i \in [N]} \{ D_i(\ul{Y}) \neq X_i \},
\end{align*}
it is easy to see that
\begin{align}
  \label{equation:erasure_prob_relations}
  P_{b,i} &\leq P_B, & P_b &\leq P_B, & P_B &\leq N P_b .
\end{align}
Also, if $D$ declares an erasure, there will be at least $\dmin$ bits in erasure.
Therefore,
\begin{align*}
  \dmin \indicator{D(\ul{Y}) \neq \ul{X}} \leq \sum_{i \in [N]} \indicator{D_i(\ul{Y}) \neq X_i} .
\end{align*}
Taking expectations on both sides gives a tighter bound on $P_B$ in terms of $P_b$,
\begin{align}
  \label{equation:erasure_prob_dmin}
  P_B \leq \frac{N}{\dmin} P_b .
\end{align}

\subsection{MAP EXIT Functions}
\label{subsection:exit_functions}
Again, let $\ul{X}=(X_1,\dots,X_N)$ denote a uniformly selected codeword from $\mc{C}$ and $\ul{Y}$ be the sequence obtained from observing $\ul{X}$ with some positions erased.
In this case, however, we assume $X_i$ is transmitted over the $\mr{BEC}(p_i)$ channel.
We refer to this as the $\mr{BEC}(\ul{p})$ channel where $\ul{p}=(p_1,\ldots,p_N)$ is the vector of channel erasure probabilities.
While one typically evaluates all quantities of interest at $\ul{p}=(p,\dots,p)$, such a parametrization provides a convenient mathematical framework for many derivations.

The vector EXIT function associated with bit $i$ of $\mc{C}$ is defined by
\begin{align*}
  h_i(\ul{p}) \triangleq \ent{X_i | \ul{Y}_{ \sim i}(\ul{p}_{\sim i})} .
\end{align*}
Also, the average vector EXIT function is defined by
\begin{align*}
  h(\ul{p}) \triangleq \frac{1}{N} \sum_{i=1}^{N} h_{i}(\ul{p}) .
\end{align*}
Note that, while we define $h_i$ as a function of $\ul{p}$ for uniformity, it does not depend on $p_i$.
In terms of vector EXIT functions, the standard scalar EXIT functions $h(p)$ and $h_i(p)$ (for $i \in [N]$) are given by
\begin{align*}
  h_i(p) &\triangleq h_i(\ul{p}) \Big{\lvert}_{\ul{p}=(p,\dots,p)}, & h(p) &\triangleq h(\ul{p}) \Big{\lvert}_{\ul{p}=(p,\dots,p)} .
\end{align*}

The bit erasure probabilities and the EXIT functions $h(p)$ and $h_i(p)$ have a close relationship.
Observe that
\begin{align*}
  \ent{X_i \mid \ul{Y}} &= \Pr(Y_i=*) \ent{X_i | \ul{Y}_{\sim i},Y_i=* } \\
  &\qquad + \Pr(Y_i=X_i) \ent{X_i | \ul{Y}_{\sim i}, Y_i = X_i} \\
  &= \Pr(Y_i=*) \ent{X_i | \ul{Y}_{\sim i}} .
\end{align*}
Therefore,
\begin{align}
\label{equation:biterasure_exit_relation}
P_{b,i}(p) &= ph_i(p), & P_b(p) &= ph(p) .
\end{align}

We now state several well-known properties of these EXIT functions~\cite{Ashikhmin-it04,Measson-it08}, which play a crucial role in the subsequent analysis.
It is worth noting that the original definition of EXIT charts in \cite{Ashikhmin-it04} focused on mutual information $I(\ul{X};\ul{Y})$ while later work on EXIT functions focused on the conditional entropy $H(\ul{X}|\ul{Y})$~\cite{Measson-it08}.
In our setting, this difference results only in trivial remappings of all discussed quantities.
\begin{proposition}
\label{proposition:vector_exit_properties_1}
For a code $\mc{C}$ on the $\mr{BEC}(\ul{p})$ channel, the EXIT function associated with bit $i$ satisfies
\begin{align*}
  h_i(\ul{p}) = \pderi{H(\ul{X}|\ul{Y}(\ul{p}))}{p_i} .
\end{align*}
For a parametrized path $\ul{p}(t)=(p_1(t),\dots,p_n(t))$ defined for $t \in [0,1]$, where $p'_i(t)$ is continuous, one finds
\begin{align*}
  \ent{\ul{X} | \ul{Y}(\ul{1})} - \ent{\ul{X} | \ul{Y}(\ul{0})} = \int_0^1 \left( \sum_{i=1}^{N} h_i(\ul{p}(t)) p'_i(t) \right) \diff{t} .
\end{align*}
\end{proposition}
\begin{IEEEproof}
This result is implied by the results of both~\cite{Ashikhmin-it04} and~\cite{Measson-it08}.
For completeness, we repeat here the proof from~\cite[Theorem 2]{Measson-it08} using our notation.

For the first statement, we start by using chain rule of entropy to write
\begin{align*}
  \ent{\ul{X} | \ul{Y}(\ul{p})} = \ent{X_i | \ul{Y}(\ul{p})} + \ent{\ul{X}_{\sim i} | X_i, \ul{Y}(\ul{p}) }.
\end{align*}
Then, we observe that 
\begin{align*}
  \ent{\ul{X}_{\sim i} | X_i, \ul{Y}(\ul{p})} = \ent{\ul{X}_{\sim i} | X_i, \ul{Y}_{\sim i}(\ul{p}_{\sim i})}, 
\end{align*}
is independent of $p_i$.
Since
\begin{align*}
  \ent{X_i | \ul{Y}(\ul{p})} &= \Pr(Y_i=*) \ent{X_i | \ul{Y}_{\sim i}(\ul{p}_{\sim i}) , Y_i=*} \\
  & \quad + \Pr(Y_i=X_i) \ent{X_i | \ul{Y}_{\sim i}(\ul{p}_{\sim i}) , Y_i=X_i} \\
  &= p_i \ent{X_i | \ul{Y}_{\sim i}(\ul{p}_{\sim i})},
\end{align*}
we find that
\begin{align*}
  \pderi{\ent{X_i | \ul{Y}(\ul{p})}}{p_i} = \ent{X_i | \ul{Y}_{\sim i}(\ul{p}_{\sim i})} = h_i(\ul{p}).
\end{align*}
The second statement now follows directly from vector calculus.
\end{IEEEproof}

The following sets characterize the EXIT functions $h_i$ and we will refer to them throughout the article.
\begin{definition}
\label{definition:Omega}
Consider a code $\mc{C}$ and the \emph{indirect recovery} of $X_i$ from the subvector $\ul{Y}_{\sim i}$ (i.e., the bit-MAP decoding of $Y_i$ from $\ul{Y}$ when $Y_i = *$).
For $i \in [N]$, the set of erasure patterns that prevent indirect recovery of $X_i$ under bit-MAP decoding is given by
  \begin{align*}
    \Omega_{i} \triangleq \Big{\{} A \subseteq [N]\backslash\{i\} \,|\, \exists B \subseteq [N] \backslash\{i\}, B \cup \{i\} \in \mc{C}, B \subseteq A \Big{\}}.
  \end{align*}  
For distinct $i,j \in [N]$, the set of erasure patterns where the $j$-th bit is \emph{pivotal} for the indirect recovery of $X_i$ is given by
  \begin{align*}
    \partial_j \Omega_i \triangleq \left\{ A \subseteq [N]\backslash \{i\} \mid A \backslash\{j\} \notin \Omega_i , A\cup\{j\} \in \Omega_i \right\}.
  \end{align*}
These are erasure patterns where $X_i$ can be recovered from $\ul{Y}_{\sim i}$ if and only if $Y_j \neq *$ (i.e., the $j$-th bit is not erased).
Note that $\partial_j \Omega_i$ includes patterns from both $\Omega_i$ and $\Omega_i^c$.

\end{definition}
Intuitively, $\Omega_i$ is the set of all erasure patterns that cover some codeword whose $i$-th bit is $1$.
Also, since the minimum distance of $\mc{C}$ is at least $2$ by assumption, the decoder can always recover bit~$i$ indirectly if no other bits are erased.
Thus, $\Omega_i$ does not contain the empty set.
For $j\in [N] \backslash i$, the set $\partial_j \Omega_i$ characterizes the boundary erasure patterns where flipping the erasure status of the $j$-th bit moves the pattern between $\Omega_i$ and $\Omega_i^c$.

\begin{proposition}
\label{proposition:vector_exit_properties_2}
For a code $\mc{C}$ on the $\mr{BEC}(\ul{p})$ channel, we have the following explicit expressions.
\begin{enumerate}[a)]
\item For bit $i$, the EXIT function is given by
  \begin{align*}
    h_{i}(\ul{p}) = \sum_{A \in \Omega_i} \prod_{\ell \in A} p_\ell \prod_{\ell \in A^c\backslash \{i\}} (1-p_\ell).
  \end{align*}
\item For distinct $i$ and $j$, the mixed partial derivative satisfies
  \begin{align*}
    \frac{\partial^2 H(\ul{X} | \ul{Y}(\ul{p}))}{\partial p_j \partial p_i} \!=\! \pderi{h_i(\ul{p})}{p_j} \!= \!\sum_{A \in \partial_j \Omega_i} \prod_{\ell \in A} p_\ell \!\!\prod_{\ell \in A^c\backslash \{i\}} (1-p_\ell).
  \end{align*}
\end{enumerate}
\end{proposition}
\begin{IEEEproof}
For a), the definition of $h_i$ implies
\begin{align*}
  h_i(\ul{p}) &= \ent{X_i | \ul{Y}_{\sim i}(p_{\sim i}) } \\
  &=\sum_{\ul{y}_{\sim i}\in \mathcal{Y}^{N-1}} \Pr(\ul{Y}_{\sim i}=\ul{y}_{\sim i}) \ent{X_i | \ul{Y}_{\sim i}=\ul{y}_{\sim i}} .
\end{align*}
Note that either $y_\ell=x_\ell$ or $y_\ell=*$.
Let $A \subseteq [N] \backslash \{i\}$ be the set of indices where $y_\ell=*$ so that
\begin{align*}
  \Pr(\ul{Y}_{\sim i}=\ul{y}_{\sim i}) = \prod_{\ell \in A} p_\ell \prod_{\ell \in A^c \backslash \{i\}} (1-p_\ell) .
\end{align*}

If $A \cup \{i\}$ covers a codeword in $\mc{C}$ whose $i$-th bit is non-zero, then bit-MAP decoder fails to decode bit $i$.
Also, since the posterior probability of $X_i$ given $\ul{Y}_{\sim i} = \ul{y}_{\sim i}$ is uniform, $H(X_i | \ul{Y}_{\sim i}=\ul{y}_{\sim i})=1$.

If $A \cup \{i\}$ does not cover any codeword in $\mc{C}$ with non-zero bit $i$, then the MAP estimate of $X_i$ given $\ul{Y}_{\sim i} = \ul{y}_{\sim i}$ is equal to $X_i$ and $H(X_i | \ul{Y}_{\sim i}=\ul{y}_{\sim i})=0$.

Thus, the EXIT function $h_i (\ul{p})$ is given by summing over the first set of erasure patterns where the entropy is $1$.
This set is precisely $\Omega_i$, the set of all erasure patterns that cover a codeword whose $i$-th bit is non-zero.

For b), we evaluate the partial derivative using the explicit evaluation of $h_i(\ul{p})$ from part a).
Suppose $A \in \Omega_i$.
To simplify things, we handle the two groups separately.

If $A\cup\{j\} \in \Omega_{i}$ and $A\backslash\{j\}  \in \Omega_i$, then we observe that
\begin{align*}
  \sum_{B\in \{A\cup\{j\},A\backslash\{j\}\}} & \;\;\, \prod_{\ell \in B} p_\ell \prod_{\ell \in B^c\backslash \{i\}} (1-p_\ell) \\
  = & \prod_{\ell \in A \backslash \{j\}} p_\ell \prod_{\ell \in A^c\backslash \{i,j\}} (1-p_\ell)
\end{align*}
is independent of the variable $p_j$.
Thus, its partial derivative with respect to $p_j$ is zero.

On the other hand, if $A \cup \{j\} \in \Omega_i$ but $A \backslash \{j\} \notin \Omega_i$, then $j \in A$.
In this case, the contribution of $A$ to $h_i(\ul{p})$ can be written as 
\begin{align*}
  h_{i,A}(\ul{p}) = \prod_{\ell \in A} p_\ell \prod_{\ell \in A^c\backslash \{i\} } (1 - p_\ell) .
\end{align*}
Since $j \in A$, we find that
\begin{align}
	\label{equation:dhia}
  \pderi{h_{i,A}(\ul{p})}{p_j} = \prod_{\ell \in A \backslash \{j\} } p_\ell \prod_{\ell \in A^c\backslash \{i\} } (1 - p_\ell) 
\end{align}
and, since the derivative is zero for patterns in the first group, we get 
\begin{align}
  \label{equation:dhip_sum}
  \pderi{h_i(\ul{p})}{p_j} \!= \! \sum_{A\in \{B \in \Omega_i  \, | B \backslash \{j\} \notin \Omega_i \} } \pderi{h_{i,A}(\ul{p})}{p_j}.
\end{align}
We can also rewrite~\eqref{equation:dhia} as
\begin{align}
	\label{equation:dhia2}
  \pderi{h_{i,A}(\ul{p})}{p_j} = \sum_{B\in \{A\cup\{j\},A\backslash\{j\}\}} \prod_{\ell \in B} p_\ell \prod_{\ell \in B^c\backslash \{i\} } (1 - p_\ell) ,
\end{align}
 where the effect of $p_j$ is removed by summing over $A \cup \{j\}$ and $A \backslash \{j\}$.
 Substituting~\eqref{equation:dhia2} into~\eqref{equation:dhip_sum} gives the desired result because
 $\partial_j \Omega_i$ is equal to the union of $\{A \in \Omega_i  \, | A \backslash \{j\} \notin \Omega_i \}$ and $\{A \not\in \Omega_i  \, | A \cup \{j\} \in \Omega_i \}$.
\end{IEEEproof}

The following proposition restates some known results in our notation.
The area theorem, stated below as c), first appeared in \cite[Theorem 1]{Ashikhmin-it04}, and the explicit evaluation of $h_i(p)$, stated below in a), is a restatement of \cite[Lemma 3.74(iv)]{RU-2008}.
\begin{proposition}
\label{proposition:scalar_exit_properties}
For a code $\mc{C}$ and transmission over a $\mr{BEC}$, we have the following properties for the EXIT functions.
\begin{enumerate}[a)]
\item The EXIT function associated with bit $i$ satisfies
  \begin{align*}
    h_{i}(p) = \sum_{A \in \Omega_{i}} p^{|A|}(1-p)^{N-1-|A|}.
  \end{align*}
\item For $j\in[N]\backslash \{i\}$, the partial derivative satisfies
  \begin{align*}
    \pderi{h_i(\ul{p})}{p_j} \Big{\lvert}_{\ul{p}=(p,\dots,p)} = \sum_{A \in \partial_j \Omega_i} p^{|A|} (1-p)^{N-1-|A|} .
  \end{align*}
\item The average EXIT function satisfies the \emph{area theorem}
  \begin{align*}
    \int_{0}^{1} h(p) \diff{p} = \frac{K}{N}.
  \end{align*}
\end{enumerate}
\end{proposition}
\begin{IEEEproof}
The first two parts follow from Proposition \ref{proposition:vector_exit_properties_2}.
For the third part, we use Proposition \ref{proposition:vector_exit_properties_1}(b) with the path $\ul{p}(t)=(t,\dots,t)$.
This gives
\begin{align*}
  \ent{\ul{X} | \ul{Y}(\ul{1})} - \ent{\ul{X} | \ul{Y}(\ul{0})} = \int_0^1 \left( \sum_{i=1}^N h_i(t) \right) \diff{t} .
\end{align*}
Also, $\ent{\ul{X} | \ul{Y}(\ul{1})} = \ent{\ul{X}} = K$ and $\ent{\ul{X} | \ul{Y}(\ul{0})} = 0$.
Combining these observations gives the desired result.
\end{IEEEproof}

Since the code $\mc{C}$ is proper by assumption, $\Omega_i$ is non-empty and, in particular, $[N]\backslash \{i\} \in \Omega_i$.
Thus, $h_i$ is not a constant function equal to $0$ and $h_i (1)=1$.
Since the minimum distance of the code $\mc{C}$ is at least $2$ by assumption, $\Omega_i$ does not contain the empty set.
This implies that $h_i$ is not a constant function equal to $1$ and that $h_i(0)=0$.
As such, $h_i$ is a non-constant polynomial.
Also, $h_i$ is non-decreasing because Proposition \ref{proposition:scalar_exit_properties}(b) implies that $dh_i/dp \geq 0$.
It follows that $h_i$ is strictly increasing because a non-constant non-decreasing polynomial must be strictly increasing.

Consequently, the EXIT functions $h_i(p)$, and therefore $h(p)$, are continuous, strictly increasing polynomial functions on $[0,1]$ with $h(0)=h_i(0)=0$ and $h(1)=h_i(1)=1$.

The inverse function for the average EXIT function is therefore well-defined on $[0,1]$.
For $t \in [0,1]$, let
\begin{align}
  \label{equation:h_inverse}
  p_t \triangleq h^{-1}(t) = \inf \{ p \in [0,1] \mid h(p) \geq t \} ,
\end{align}
and note that $h(p_t)=t$.

\subsection{Permutations of Linear Codes}
\label{subsection:permutations}

Let $S_N$ be the symmetric group on $N$ elements.
The permutation group of a code is defined as the subgroup of $S_N$ whose group action on the bit ordering preserves the set of codewords~\cite[Section 1.6]{Huffman-2003}.

\begin{definition}
The permutation group $\mc{G}$ of a code $\mc{C}$ is defined to be
\begin{align*}
  \mc{G} = \left\{ \pi \in S_N \mid \text{$\pi(A) \in \mc{C}$ for all $A \in \mc{C}$} \right\} .
\end{align*}
\end{definition}

Interestingly, for binary linear codes, the permutation group is isomorphic to the group of weight-preserving linear transformations of the code \cite[Section 8.5]{Macwilliams-1977}, \cite[Section 7.9]{Huffman-2003}, \cite{Berger-it96}.

\begin{definition}
Suppose $\mc{G}$ is a permutation group. Then, 
\begin{enumerate}[a)]
\item  $\mc{G}$ is \emph{transitive} if, for any $i,j \in [N]$, there exists a permutation $\pi \in \mc{G}$ such that $\pi(i) = j$, and
\item $\mc{G}$ is \emph{doubly transitive} if, for any distinct $i,j,k \in [N]$, there exists a $\pi \in \mc{G}$ such that $\pi(i) = i$ and $\pi(j)=k$. 
\end{enumerate}
\end{definition}

Note that any non-trivial code (i.e., $0<r<1$) whose permutation group is transitive must be proper and have minimum distance at least two.

In the following, we explore some interesting symmetries of EXIT functions when the permutation group of the code is transitive or doubly transitive.

\begin{proposition}
\label{proposition:transitive_exit}
Suppose the permutation group $\mc{G}$ of a code $\mc{C}$ is transitive.
Then, for any $i \in [N]$,
\begin{align*}
  h(p) = h_i(p) \qquad \text{for $0\leq p \leq 1$} .
\end{align*}
\end{proposition}
\begin{IEEEproof}
Since $\mc{G}$ is transitive, for any $i,j \in [N]$, there exists a permutation $\pi$ such that $\pi(i)=j$.
This will allow us to show that there is a bijection between $\Omega_i$ and $\Omega_j$ induced by the action of $\pi$ on the codeword indices.
To do this, we first show that $A \in \Omega_i$ implies $\pi(A) \in \Omega_j$.

Since $A \in \Omega_i$, by definition, there exists $B \subseteq A$ such that $B \cup \{ i \} \in \mc{C}$.
Since $\pi \in \mc{G}$, $\pi(B \cup \{i\}) \in \mc{C}$.
Also, $\pi(B \cup \{i\}) = \pi(B) \cup \{j\}$ and $\pi(B) \subseteq \pi(A)$.
Consequently, $\pi(A) \in \Omega_j$.

Similarly, if $A \in \Omega_j$, then $\pi^{-1}(A) \in \Omega_i$.
Thus, there is a bijection between $\Omega_i$ and $\Omega_j$ induced by $\pi$.
This bijection also preserves the weight of the vectors in each set (i.e., $|A|=|\pi(A)|$).

Since Proposition \ref{proposition:scalar_exit_properties}(a) implies that $h_i(p)$ only depends on the weights of elements in $\Omega_i$, it follows that $h_i(p)=h_j(p)$.
This also implies that $h(p)=h_i(p)$ for all $0 \leq p \leq 1$.
\end{IEEEproof}

\begin{proposition}
\label{proposition:2transitive_exit_deri}
Suppose that the permutation group $\mc{G}$ of a code $\mc{C}$ is doubly transitive.
Then, for distinct $i,j,k \in [N]$, and any $0 \leq p \leq 1$,
\begin{align*}
  \pderi{h_i(\ul{p})}{p_j} \Big{\lvert}_{\ul{p}=(p,\dots,p)} = \pderi{h_i(\ul{p})}{p_k} \Big{\lvert}_{\ul{p}=(p,\dots,p)} .
\end{align*}
\end{proposition}
\begin{IEEEproof}
Since $\mc{G}$ is doubly transitive, there exists a permutation $\pi \in \mc{G}$ such that $\pi(i)=i$ and $\pi(j)=k$.
Suppose $A \in \partial_j \Omega_i$.
Then, by definition, either 1) $A \in \Omega_i$ and $A \backslash \{j\} \notin \Omega_i$ or 2) $A \cup \{j\} \in \Omega_i$ and $A \notin \Omega_i$.
In either case, we claim that $\pi(A) \in \partial_k \Omega_i$.
We prove this for the first case.
The proof for the second case can be obtained verbatim by replacing $A$ with $A \cup \{j\}$.

Suppose $A \in \Omega_i$ and $A \backslash \{j\} \notin \Omega_i$.
Since $\pi \in \mc{G}$ and $\pi(i)=i$, $\pi(A) \in \Omega_i$.
Also, $\pi(A \backslash \{j\}) \notin \Omega_i$; otherwise, $A \backslash\{j\}=\pi^{-1}(\pi(A \backslash \{ j \})) \in \Omega_i$ gives a contradiction.
Finally, $\pi(A \backslash \{j\}) = \pi(A) \backslash \{k\}$ implies that $\pi(A) \in \partial_k \Omega_i$.
Similarly, one finds that $A \in \partial_k \Omega_i$ implies $\pi^{-1}(A) \in \partial_j \Omega_i$.

Since Proposition \ref{proposition:scalar_exit_properties}(b) implies that $\tfrac{\partial h_i}{\partial p_j}|_{\ul{p}=(p,\dots,p)}$ only depends on the weights of elements in $\partial_j \Omega_i$ and $|A| = |\pi(A)|$, we obtain the desired result.
\end{IEEEproof}

\subsection{Capacity-Achieving Codes}
\label{subsection:capacity}

\begin{figure}[!t]
	\centering
	\setlength\tikzheight{5.5cm}
	\setlength\tikzwidth{7cm} 
%
%
\begin{tikzpicture}

\begin{axis}[%
scale only axis,
width=\tikzwidth,
height=\tikzheight,
xmin=0, xmax=1,
ymin=0, ymax=1,
xlabel={Erasure Probability},
ylabel={Average EXIT Function $h$},
xtick={0,0.25,0.5,0.75,1},
ytick={0,0.25,0.5,0.75,1},
xmajorgrids,
ymajorgrids,
legend pos=south east,
legend entries={$N=2^3$,$N=2^5$,$N=2^7$,$N=2^9$},
legend style={nodes=right}]
axis on top]

\addplot [
color=black,
solid,
line width=1.25pt,
]
coordinates{
 (0,0)(0.01,1e-05)(0.02,7e-05)(0.03,0.00022)(0.04,0.000415)(0.05,0.000895)(0.06,0.00141)(0.07,0.00229)(0.08,0.00356)(0.09,0.00479)(0.1,0.006575)(0.11,0.009055)(0.12,0.011475)(0.13,0.015265)(0.14,0.018335)(0.15,0.022405)(0.16,0.027035)(0.17,0.032355)(0.18,0.037055)(0.19,0.04472)(0.2,0.050475)(0.21,0.058145)(0.22,0.06635)(0.23,0.07406)(0.24,0.083545)(0.25,0.093085)(0.26,0.103875)(0.27,0.11443)(0.28,0.127385)(0.29,0.13941)(0.3,0.1531)(0.31,0.1643)(0.32,0.17898)(0.33,0.191685)(0.34,0.209565)(0.35,0.22444)(0.36,0.240355)(0.37,0.25538)(0.38,0.274675)(0.39,0.290965)(0.4,0.30721)(0.41,0.32741)(0.42,0.346235)(0.43,0.366955)(0.44,0.382625)(0.45,0.40333)(0.46,0.420825)(0.47,0.44083)(0.48,0.46153)(0.49,0.480005)(0.5,0.501335)(0.51,0.518265)(0.52,0.53818)(0.53,0.557755)(0.54,0.577755)(0.55,0.596545)(0.56,0.61576)(0.57,0.63548)(0.58,0.65362)(0.59,0.672805)(0.6,0.6892)(0.61,0.70715)(0.62,0.72629)(0.63,0.742375)(0.64,0.75863)(0.65,0.775575)(0.66,0.791245)(0.67,0.806415)(0.68,0.820405)(0.69,0.834655)(0.7,0.84767)(0.71,0.85982)(0.72,0.871845)(0.73,0.883675)(0.74,0.895455)(0.75,0.906055)(0.76,0.916075)(0.77,0.925055)(0.78,0.93421)(0.79,0.94192)(0.8,0.94957)(0.81,0.95568)(0.82,0.962365)(0.83,0.968405)(0.84,0.972825)(0.85,0.978035)(0.86,0.982435)(0.87,0.98551)(0.88,0.98816)(0.89,0.991105)(0.9,0.993215)(0.91,0.995125)(0.92,0.996405)(0.93,0.99763)(0.94,0.99862)(0.95,0.99904)(0.96,0.99958)(0.97,0.999855)(0.98,0.999925)(0.99,0.99999)(1,1) 
};

\addplot [
color=blue,
solid,
line width=1.25pt,
]
coordinates{
(0,0)(0.1,0)(0.11,2e-05)(0.12,2e-05)(0.13,4e-05)(0.14,0.00012)(0.15,0.0004)(0.16,0.00038)(0.17,0.00072)(0.18,0.00094)(0.19,0.00148)(0.2,0.00184)(0.21,0.00242)(0.22,0.00366)(0.23,0.00524)(0.24,0.00672)(0.25,0.00938)(0.26,0.01146)(0.27,0.01534)(0.28,0.01964)(0.29,0.02462)(0.3,0.0313)(0.31,0.03988)(0.32,0.04772)(0.33,0.05492)(0.34,0.0706)(0.35,0.08308)(0.36,0.09892)(0.37,0.11364)(0.38,0.13566)(0.39,0.15612)(0.4,0.18202)(0.41,0.20542)(0.42,0.23124)(0.43,0.2612)(0.44,0.29242)(0.45,0.32166)(0.46,0.35662)(0.47,0.3916)(0.48,0.42728)(0.49,0.46054)(0.5,0.4979)(0.51,0.53438)(0.52,0.57176)(0.53,0.60832)(0.54,0.6396)(0.55,0.67108)(0.56,0.70504)(0.57,0.73844)(0.58,0.76682)(0.59,0.79678)(0.6,0.8208)(0.61,0.84448)(0.62,0.8628)(0.63,0.88374)(0.64,0.90152)(0.65,0.91588)(0.66,0.92902)(0.67,0.9404)(0.68,0.95138)(0.69,0.96136)(0.7,0.96862)(0.71,0.97486)(0.72,0.97936)(0.73,0.984)(0.74,0.98724)(0.75,0.99062)(0.76,0.9925)(0.77,0.99498)(0.78,0.9962)(0.79,0.99718)(0.8,0.99836)(0.81,0.99836)(0.82,0.9993)(0.83,0.99934)(0.84,0.9995)(0.85,0.99972)(0.86,0.99974)(0.87,0.99994)(0.88,0.99992)(0.89,0.99996)(0.9,1)(0.91,0.99998)(0.92,1)(1,1)
};

\addplot [
color=red,
solid,
line width=1.25pt,
]
coordinates{
(0,0)(0.24,0)(0.25,2e-05)(0.26,0)(0.27,8e-05)(0.28,8e-05)(0.29,0.00012)(0.3,0.00028)(0.31,0.00032)(0.32,0.00038)(0.33,0.0013)(0.34,0.0019)(0.35,0.0024)(0.36,0.0043)(0.37,0.00628)(0.38,0.0111)(0.39,0.01728)(0.4,0.02716)(0.41,0.03968)(0.42,0.05848)(0.43,0.08578)(0.44,0.12144)(0.45,0.16572)(0.46,0.21808)(0.47,0.27896)(0.48,0.34492)(0.49,0.42268)(0.5,0.50242)(0.51,0.57696)(0.52,0.65078)(0.53,0.72416)(0.54,0.78204)(0.55,0.83668)(0.56,0.87948)(0.57,0.91572)(0.58,0.94088)(0.59,0.96014)(0.6,0.97456)(0.61,0.98334)(0.62,0.98892)(0.63,0.9934)(0.64,0.99618)(0.65,0.99722)(0.66,0.99862)(0.67,0.9991)(0.68,0.99948)(0.69,0.99968)(0.7,0.99964)(0.71,0.9999)(0.72,0.99992)(0.73,0.99996)(0.74,0.99998)(0.75,1)(1,1)
};

\addplot [
color=olive!50!black,
solid,
line width=1.25pt,
]
coordinates{
(0,0)(0.41,0)(0.42,0.00025)(0.43,0.00175)(0.44,0.00275)(0.45,0.01525)(0.46,0.03725)(0.47,0.08975)(0.48,0.179)(0.49,0.33925)(0.5,0.49925)(0.51,0.653)(0.52,0.80025)(0.53,0.90925)(0.54,0.9595)(0.55,0.98225)(0.56,0.9955)(0.57,0.999)(0.58,0.99975)(0.59,1)(1,1)
};

\end{axis}
\end{tikzpicture}

	\caption{The average EXIT function of the rate-$1/2$ Reed-Muller code with blocklength $N$.}
	\label{figure:reed_muller_exit_simulation}
\end{figure}
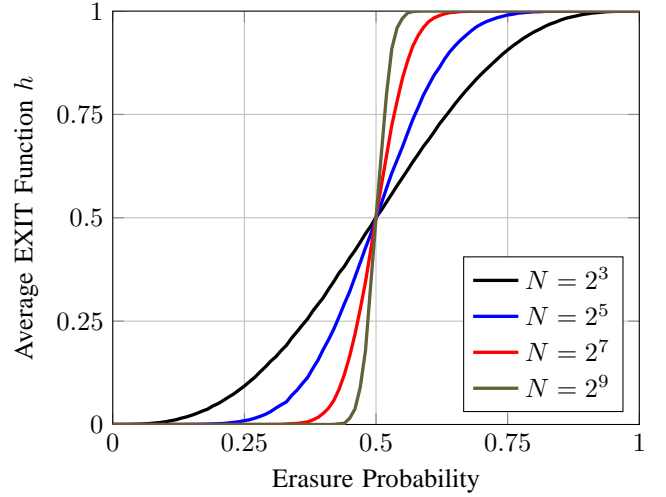

\begin{definition}
Suppose $\{ \mc{C}_n \}$ is a sequence of codes with rates $\{r_n\}$ where $r_n \to r$ for $r\in (0,1)$.
\begin{enumerate}[a)]
\item $\{ \mc{C}_n \}$ is said to be capacity achieving on the BEC under bit-MAP decoding, if for any $p \in [0,1-r)$, the average bit-erasure probabilities satisfy
  \begin{align*}
    \lim_{n \to \infty} P_b^{(n)}(p) = 0 .
  \end{align*}
\item $\{ \mc{C}_n \}$ is said to be capacity achieving on the BEC under block-MAP decoding, if for any $p \in [0,1-r)$, the block-erasure probabilities satisfy
  \begin{align*}
    \lim_{n \to \infty} P_B^{(n)}(p) = 0 .
  \end{align*}
\end{enumerate}
\end{definition}
Note that in the definition above, we do not impose any constraints on the blocklength of the code $\mc{C}_n$.

The following proposition encapsulates our approach in showing capacity achievability.
It naturally bridges capacity-achieving codes, average EXIT functions, and the sharp transition framework presented in the next section that allows us to show that the transition width\footnote{Defined as the interval length where the function goes from $\epsilon$ to $1-\epsilon$.} of certain functions goes to $0$.
The average EXIT functions of some rate-$1/2$ Reed-Muller codes are shown in Figure~\ref{figure:reed_muller_exit_simulation}.
Observe that as the blocklength increases, the transition width of the average EXIT function decreases.
According to the following proposition, if this width converges to $0$, then Reed-Muller codes achieve capacity on the BEC under bit-MAP decoding. 

\begin{proposition}
\label{proposition:capacity_equivalence}
Let $\{\mc{C}_n\}$ be a sequence of codes with rates $\{r_n\}$ where $r_n \to r$ for $r\in (0,1)$.
Then, the following statements are equivalent.
\begin{enumerate}[S1:]
\item $\{ \mc{C}_n \}$ is capacity achieving on the BEC under bit-MAP decoding.
\item The sequence of average EXIT functions satisfies
  \begin{align*}
    \lim_{n \to \infty} h^{(n)} (p) =
    \begin{cases}
      0 & \text{if $0 \leq p<1-r$}, \\
      1 & \text{if $1-r<p \leq 1$}. 
    \end{cases}
  \end{align*}
\item For any $0< \epsilon \leq 1/2$,
  \begin{align*}
    \lim_{n \to \infty} \left( p_{1-\epsilon}^{(n)} - p_{\epsilon}^{(n)} \right) = 0 ,
  \end{align*}
  where $p_t^{(n)}$ is the functional inverse of $h^{(n)}$ given by \eqref{equation:h_inverse}.
\end{enumerate}
\end{proposition}
\begin{IEEEproof}
See Appendix \ref{appendix:proof_capacity_equivalence}.

The equivalence between the first two statements is due to the close relationship between the bit erasure probability and the average EXIT function in \eqref{equation:biterasure_exit_relation}, while the equivalence between the last two statements is a consequence of the area theorem in Proposition \ref{proposition:scalar_exit_properties}(c).
\end{IEEEproof}

While the above result appears deceptively simple, our approach is successful largely because the transition point of the limiting EXIT function is known a priori due to the area theorem.
Even though the sharp transition framework presented in the next section is widely applicable in theoretical computer science and allows one to deduce that the transition width of certain functions goes to $0$, establishing the existence of a threshold and determining its precise location if it exists can be notoriously difficult \cite{Achlioptas-nature05,Coja-Oghlan-stoc14,Ding-arxiv14}.


\section{Sharp Thresholds for Monotone Boolean Functions via Isoperimetric Inequalities}
\label{section:sharp_threshold}
As seen in Proposition \ref{proposition:capacity_equivalence}, the crucial step in proving the achievability of capacity for a sequence of codes is to show that the average EXIT function transitions sharply from $0$ to $1$.
From the explicit evaluation of $h_i$ in Proposition \ref{proposition:scalar_exit_properties}(a), it is clear that the set $\Omega_i$ defines the behavior of $h_i$.
Indeed, these sets play a crucial role in our analysis.

In this section, we treat the sets $\Omega_i$ and $\partial_j \Omega_i$ from Definition~\ref{definition:Omega} as a set of sequences in $\{0,1\}^{N-1}$, since index $i$ is not present in any of their elements.
This occurs because $h_i (\ul{p})$ is not a function of $p_i$.
To make this notion precise, we associate $A \subseteq [N] \backslash \{i\}$ with $\Phi_i(A) \in \{0,1\}^{N-1}$, where bit $\ell$ of $\Phi_i(A)$ is given by
\begin{align*}
  [\Phi_i(A)]_{\ell} \triangleq
  \begin{cases}
    \mathbbm{1}_A (\ell) & \text{if $\ell < i$}, \\
    \mathbbm{1}_A (\ell+1) & \text{if $\ell \geq  i$} .
  \end{cases}
\end{align*}
Now, define
\begin{align}
  \label{equation:Omega_prime}
  \Omega'_i &\triangleq \{ \Phi_i(A) \in \{0,1\}^{N-1} \mid A \in \Omega_i \}, \\
  \partial_j \Omega'_i &\triangleq \{ \Phi_i(A) \in \{0,1\}^{N-1} \mid A \in \partial_j \Omega_i \} . \nonumber
\end{align}
Whenever we treat $\Omega_i$ and $\partial_j \Omega_i$ as sequences of length $N-1$, we refer to them as $\Omega'_i$ and $\partial_j \Omega'_i$ to avoid confusion.

Consider the space $\{0,1\}^{M}$ with a measure $\mu_p$ such that
\begin{align*}
  \mu_p(\Omega) &= \sum_{\ul{x} \in \Omega} p^{|\ul{x}|} (1-p)^{M-|\ul{x}|}, \qquad \text{for $\Omega \subseteq \{0,1\}^{M}$},
\end{align*}
where the weight $|\ul{x}|=x_1+\dots+x_{M}$ is the number of $1$'s in $\ul{x}$.
We note that $h_i(p) = \mu_p(\Omega'_i)$ with $M=N-1$.

Recall that for $\ul{x},\ul{y} \in \{0,1\}^{M}$, we write $\ul{x} \leq \ul{y}$ if $x_i \leq y_i$ for all $i \in [M]$.
\begin{definition}
A non-empty proper subset $\Omega \subset \{0,1\}^{M}$ is called \emph{monotone} if $\ul{x} \in \Omega$ and $\ul{x} \leq \ul{y}$, then $\ul{y} \in \Omega$.
\end{definition}

\begin{remark}
If the bit-MAP decoder cannot recover bit $i$ from a received sequence, then it cannot recover bit $i$ from any received sequence formed by adding additional erasures to the original received sequence.
This implies that the set $\Omega'_i$ is monotone.
\end{remark}

Monotone sets appear frequently in the theory of random graphs, satisfiability problems, etc.
For a monotone set $\Omega$, $\mu_p(\Omega)$ is a strictly increasing function of $p$.
Often, the quantity $\mu_p(\Omega)$ exhibits a threshold type behavior, as a function of $p$, where it jumps quickly from $0$ to $1$.
One technique that has been surprisingly effective in showing this behavior is based on deriving inequalities of the form
\begin{align}
  \label{equation:differential_inequality}
  \frac{\diff{\mu_p(\Omega)}}{\diff{p}} \geq w \mu_p(\Omega) (1-\mu_p(\Omega)) .
\end{align}
If $w$ is large, then the derivative of $\mu_p(\Omega)$ will be large when $\mu_p(\Omega)$ is not close to either $0$ or $1$.
In this case, $\mu_p(\Omega)$ must transition from $0$ to $1$ over a narrow range of $p$ values.

One elegant way to obtain such inequalities is based on discrete isoperimetric inequalities~\cite{Kalai-05,Boucheron-2013}.
First, let us define the function $g_{\Omega} \colon \{0,1\}^{M} \to \mathbb{N} \cup \{0\}$, which quantifies the boundary between $\Omega$ and $\Omega^c$,
\begin{align} \label{equation:boundary}
  g_{\Omega}(\ul{x}) \triangleq
  \begin{cases}
    \Big{\lvert}\{ \ul{y} \in \Omega^c \mid d_{\mathrm{H}}(\ul{x},\ul{y}) = 1 \}\Big{\rvert} & \text{if $\ul{x} \in \Omega$}, \\
    0 & \text{if $\ul{x} \notin \Omega$} ,
  \end{cases}
\end{align}
where $d_{\mathrm{H}}$ is the Hamming distance.
Surprisingly, for a monotone set $\Omega$, the derivative $\diff{\mu_p(\Omega)}/\diff{p}$ can be characterized exactly by $g_{\Omega}$ according to the Margulis-Russo Lemma \cite{Margulis-ppi74,Russo-prf82}:
\begin{align*}
  \frac{\diff{\mu_p(\Omega)}}{\diff{p}} = \frac{1}{p} \int g_{\Omega}(\ul{x}) \mu_p(\diff{\ul{x}}) .
\end{align*}

Observe that $\mu_p(\Omega) + \mu_p(\Omega^c) = 1$ for any $0 \leq p \leq 1$.
For a monotone set $\Omega$, as we increase $p$, the probability from $\Omega^c$ flows to $\Omega$.
Intuitively, Margulis-Russo Lemma says that this flow depends only on the boundary between $\Omega$ and $\Omega^c$.
To obtain inequalities of type \eqref{equation:differential_inequality}, one approach is to find a lower bound on $g_{\Omega}$ that holds whenever it is non-zero~\cite{Margulis-ppi74,Russo-prf82}.

These techniques were introduced to coding by Tillich and Z{\'e}mor to analyze the block error rate of linear codes under block-MAP decoding~\cite{Zemor-94,Tillich-cpc00}.
In that case, the minimum non-zero value of $g_{\Omega}$ is proportional to the minimum distance of the code. 
Unfortunately, for the bit-MAP decoding problem we consider, the minimum non-zero $g_{\Omega}$ may be small (e.g., 1) even when the minimum distance of the code is arbitrarily large.
This is discussed further in Section~\ref{section:TillichZemor}.
To circumvent this, we discuss another approach, which requires a different formulation of the Margulis-Russo Lemma.
We begin with a few definitions.

\begin{definition}
Let $\Omega$ be a monotone set and let
\begin{align*}
  \partial_j \Omega \triangleq \left\{ \ul{x} \in \{0,1\}^{M} \mid  \mathbbm{1}_{\Omega}(\ul{x}) \neq \mathbbm{1}_{\Omega}(\ul{x}^{(j)})  \right\} ,
\end{align*}
where $\ul{x}^{(j)}$ is defined by $x^{(j)}_\ell=x_{\ell}$ for $\ell \neq j$ and $x^{(j)}_j=1-x_j$.
Let the \emph{influence of bit $j \in [M]$}
 be defined by
\begin{align*}
  I_{j}^{(p)}(\Omega) \triangleq \mu_p\left( \partial_j \Omega \right)
\end{align*}
and the \emph{total influence}
 be defined by
\begin{align*}
  I^{(p)}(\Omega) \triangleq \sum_{\ell=1}^{M} I^{(p)}_{\ell}(\Omega) .
\end{align*}
\end{definition}

The Margulis-Russo Lemma can also be stated in terms of the total influence.
\begin{theorem}[{\hspace{-0.01cm}\cite[Theorem 9.15]{Boucheron-2013}}]
\label{theorem:margulis_russo}
Let $\Omega$ be a monotone set.
Then,
\begin{align*}
  \frac{\diff{\mu_p(\Omega)}}{\diff{p}} = I^{(p)}(\Omega) .
\end{align*}
\end{theorem}

\begin{remark}
Note that we have already seen Theorem \ref{theorem:margulis_russo} in the context of EXIT functions.
When $M=N-1$, it is easy to see from Proposition \ref{proposition:scalar_exit_properties} that
\begin{align*}
  h_i(p) &= \mu_p(\Omega'_i),  & I_j^{(p)}(\Omega'_i) = \pderi{h_i(\ul{p})}{p_{j'}} \Big{\lvert}_{\ul{p}=(p,\dots,p)} ,
\end{align*}
where\vspace{-0.2cm}
\begin{align}
\label{equation:jprime}
  j' =
  \begin{cases}
    j & \text{if $j < i$}, \\
    j+1 & \text{if $j \geq i$}.
  \end{cases}
\end{align}
Therefore, Theorem \ref{theorem:margulis_russo} is equivalent to
\begin{align*}
  \deri{h_i(p)}{p} = \sum_{j \in [N]\backslash\{i\}} \pderi{h_i(\ul{p})}{p_j} \Big{\lvert}_{\ul{p}=(p,\dots,p)} ,
\end{align*}
a straightforward result from vector calculus since $h_i$ does not depend on $p_i$.
\end{remark}

The advantage of using influences over $g_{\Omega}$ is that with ``sufficient symmetry'' in $\Omega$, it is possible to show threshold phenomenon without any other knowledge about $\Omega$.
The following theorem illustrates the power of symmetry.
Our proof hinges on this result.
\begin{theorem}
\label{theorem:width_bound}
Let $\Omega$ be a monotone set and
suppose that, for all $0 \leq p \leq 1$, the influences of all bits are equal $I_1^{(p)}(\Omega)=\dots=I_M^{(p)}(\Omega)$.
\begin{enumerate}[a)]
\item Then, there exists a universal constant $C \geq 1$, which is independent of $p$, $\Omega$, and $M$, such that
  \begin{align*}
    \deri{\mu_p(\Omega)}{p} \geq C (\log M)\mu_p(\Omega)(1-\mu_p(\Omega)) ,
  \end{align*}
for all $0<p<1$.
\item
Consequently, for any $0< \epsilon \leq 1/2$,
\begin{align*}
  p_{1-\epsilon} - p_{\epsilon} \leq \frac{2}{C} \frac{\log \frac{1-\epsilon}{\epsilon}}{ \log M} ,
\end{align*}
where $p_t = \inf \{ p \in [0,1] \mid \mu_p(\Omega) \geq t\}$ is well-defined because $\mu_p (\Omega)$ is strictly increasing in $p$ with $\mu_{0}(\Omega)=0$ and $\mu_{1}(\Omega)=1$.
\end{enumerate}
\end{theorem}
\begin{IEEEproof}
See \cite{Talagrand-ap94,Friedgut-procams96,Rossignol-ap06}, \cite[Section 9.6]{Boucheron-2013} for details.

In this form (i.e., by assuming all influences are equal), this result first appeared in \cite{Friedgut-procams96}.
However, this theorem can be seen as an immediate consequence of the earlier result in \cite[Corollary 1.4]{Talagrand-ap94}.
The constant $C$ was later improved in \cite{Rossignol-ap06}.
From the outline in \cite[Exercise 9.8]{Boucheron-2013}, one can verify this theorem for $C=1$.

For the historical context, the study of influences for boolean functions was initiated in a 1987 technical report that led to~\cite{Ben-Or-rand90}.
Shortly after, \cite{Kahn-focs88} applied harmonic analysis to obtain some powerful general theorems about boolean functions.
These results were subsequently generalized in \cite{Talagrand-ap94}.
\end{IEEEproof}

For the sets $\Omega'_i$, such a symmetry between influences is imposed by the doubly transitive property of the permutation group of the code according to Proposition \ref{proposition:2transitive_exit_deri}.

\begin{theorem}
\label{theorem:2transitive_capacity}
Let $\{\mc{C}_n\}$ be a sequence of codes where the blocklengths satisfy $N_n \to \infty$, the rates satisfy $r_n \to r$, and the permutation group $\mc{G}^{(n)}$ (of $\mc{C}_n$) is doubly transitive for each $n$.
If $r\in (0,1)$, then $\{\mc{C}_n\}$ is capacity achieving on the BEC under bit-MAP decoding.
\end{theorem}
\begin{IEEEproof}
Let the average EXIT function of $\mc{C}_n$ be $h^{(n)}$.
The quantities $N$, $\mc{G}$, $h$, $h_i$, $\Omega'_i$, and $p_t$ that appear in this proof are all indexed by $n$; we drop the index to avoid cluttering.
Fix some $i \in [N]$.
Since $\mc{G}$ is transitive, from Proposition \ref{proposition:transitive_exit},
\begin{align*}
  h(p) = h_i(p), \quad \text{for all $p \in [0,1]$.}
\end{align*}
Consider the sets $\Omega'_i$ from Definition \ref{definition:Omega} and \eqref{equation:Omega_prime}, and let $M=N-1$.
Observe that, from Proposition \ref{proposition:scalar_exit_properties},
\begin{align*}
  h_i(p) &= \mu_p(\Omega'_i), &  I_j^{(p)}(\Omega'_i) &= \pderi{h_i(\ul{p})}{p_{j'}} \Big{\lvert}_{\ul{p}=(p,\dots,p)} ,
\end{align*}
where $j'$ is given in \eqref{equation:jprime}.
Since $\mc{G}$ is doubly transitive, from Proposition \ref{proposition:2transitive_exit_deri},
\begin{align*}
  I_j^{(p)}(\Omega'_i) = I_k^{(p)}(\Omega'_i) \quad \text{for all $j,k \in [N-1]$}.
\end{align*}
Using Theorem \ref{theorem:width_bound}, we have
\begin{align*}
  p_{1-\epsilon} - p_{\epsilon} \leq \frac{2}{C} \frac{\log \frac{1-\epsilon}{\epsilon}}{ \log (N-1)} ,
\end{align*}
where $p_t$ is the functional inverse of $h$ from \eqref{equation:h_inverse}.
Since $N \to \infty$ from the hypothesis,
\begin{align*}
  \lim_{n \to \infty} \left( p_{1-\epsilon} - p_{\epsilon} \right) = 0.
\end{align*}
Therefore, from Proposition \ref{proposition:capacity_equivalence}, $\{\mc{C}_n\}$ is capacity achieving on the BEC under bit-MAP decoding.
\end{IEEEproof}

We now focus on the block erasure probability.
Recall from \eqref{equation:erasure_prob_relations} and \eqref{equation:erasure_prob_dmin} that the block erasure probability satisfies the upper bounds
\begin{align*}
  P_B &\leq \frac{N P_b}{\dmin}, & P_B &\leq N P_b.
\end{align*}
Thus, if $P_b\to 0$ with sufficient speed, then $P_B \to 0$ as well.
Using \eqref{equation:differential_inequality}, one can derive the upper bound (see Lemma~\ref{lemma:exit_integral_generic} in Appendix \ref{appendix:proofs_section_inequalities} for a proof)
\begin{align*}
  \mu_{\delta}(\Omega) \leq \exp \left( -w [p_{1/2} - \delta] \right),
\end{align*}
where $p_{1/2} \in [0,1]$ is defined uniquely by $\mu_{p_{1/2}}(\Omega)=1/2$.

For $p<1-r$, the factor of $\log(N-1)$ in Theorem~\ref{theorem:2transitive_capacity} determines the decay rate of $h$ with $N$, and consequently the decay rate of $P_b$.
The following theorem shows that, if $\dmin$ satisfies $\log(\dmin)/\log(N) \to 1$, then this decay rate is also sufficient to show that $P_B \to 0$.

\begin{theorem}
\label{theorem:capacity_blockerror_dmin}
Let $\{\mc{C}_n\}$ be a sequence of codes where the blocklengths satisfy $N_n \to \infty$ and the rates satisfy $r_n \to r$ for $r\in (0,1)$.
Suppose that the average EXIT function of $\mc{C}_n$ also satisfies, for $0<p<1$,
\begin{align*}
  \deri{h^{(n)}(p)}{p} \geq  C \log (N_n) h^{(n)}(p) (1-h^{(n)}(p)),
\end{align*}
where $C>0$ is a constant independent of $p$ and $n$.
If the minimum distances $\{ \dmin^{(n)} \}$ satisfy
\begin{align*}
	\lim_{n \to \infty}  \frac{\log \dmin^{(n)}}{\log N_n} = 1,
\end{align*}
then $\{\mc{C}_n\}$ is capacity achieving on the BEC under block-MAP decoding.
\end{theorem}
\begin{IEEEproof}
  See Appendix \ref{appendix:proof_capacity_blockerror_dmin}.
\end{IEEEproof}

If $\dmin$ does not grow rapidly enough (e.g., sequences of Reed-Muller codes with rates $r_n \to r \in (0,1)$ have $\dmin = O(\sqrt{N}^{1+\delta})$ for any $\delta>0$), then the previous theorem does not apply.
Fortunately, it is possible to exploit symmetries, beyond the double transitivity of the permutation group, to obtain inequalities like \eqref{equation:differential_inequality} that grow asymptotically faster than $\log(N)$~\cite{Bourgain-gafa97}.
In particular, one obtains inequalities of type \eqref{equation:differential_inequality}, with factors of higher order than $\log(N)$, for all $p$ except a neighborhood around $0$ and $1$ that vanishes as $N\to \infty$.
The following theorem shows that this is sufficient to show that $P_B \to 0$ without imposing requirements on $\dmin$.

\begin{theorem}
\label{theorem:capacity_blockerror_wn}
Let $\{\mc{C}_n\}$ be a sequence of codes where the blocklengths satisfy $N_n \to \infty$ and the rates satisfy $r_n \to r$ for $r\in (0,1)$.
Suppose that the average EXIT function of $\mc{C}_n$ also satisfies, for $a_n<p<b_n$,
\begin{align*}
  \deri{h^{(n)}(p)}{p} \geq  w_n \log (N_n) h^{(n)}(p) (1-h^{(n)}(p)),
\end{align*}
where  $w_n \to \infty$, $a_n \to 0$, $b_n \to 1$ and $0 \leq a_n < b_n \leq 1$.
Then, $\{\mc{C}_n\}$ is capacity achieving on the BEC under block-MAP decoding.
\end{theorem}
\begin{IEEEproof}
See Appendix \ref{appendix:proof_capacity_blockerror_wn}.
\end{IEEEproof}


\section{Applications}
\label{section:applications}

\subsection{Affine-Invariant Codes}
\label{subsection:affine_invariant_codes}

Consider a code $\mc{C}$ of length $N=2^m$ and the Galois field $\mbb{F}_{N}$.
Let $\Theta \colon [N] \to \mbb{F}_N$ denote a bijection between the elements of the field and the code bits.
Take a pair $\beta, \gamma \in \mbb{F}_{N}$ with $\beta \neq 0$ and define $\pi_{\beta,\gamma} \in S_{N}$ such that
\begin{align*}
  \pi_{\beta,\gamma}(\ell)=\Theta^{-1}( \beta \Theta(\ell) + \gamma ) .
\end{align*}
Note that $\pi_{\beta,\gamma}$ is well-defined since $\Theta$ is bijective and $\beta \neq 0$, and observe that $\pi_{\beta_1,\gamma_1} \circ \pi_{\beta_2,\gamma_2} = \pi_{\beta_1\beta_2,\beta_1 \gamma_2 + \gamma_1}$.
As such, the collection of permutations $\pi_{\beta,\gamma}$ forms a group.
Now, the code $\mc{C}$ is called \emph{affine-invariant} if its permutation group contains the subgroup
\begin{align*}
  \{ \pi_{\beta,\gamma} \in S_N \mid \beta,\gamma \in \mbb{F}_{N}, \beta \neq 0 \} ,
\end{align*}
for some bijection $\Theta$ \cite[Section 4.7]{Huffman-2003}.

Affine-invariant codes are of interest to us because their permutation groups are doubly transitive.
To see this, consider distinct $i,j,k \in [N]$ and choose $\beta, \gamma \in \mbb{F}_{N}$ where
\begin{align*}
  \beta &= \frac{\Theta(i) - \Theta(k)}{\Theta(i) - \Theta(j)}, & \gamma &= \Theta(i) \left( \frac{\Theta(k) - \Theta(j)}{\Theta(i)-\Theta(j)} \right) ,
\end{align*}
and observe that $\pi_{\beta,\gamma}(i)=i$ and $\pi_{\beta,\gamma}(j)=k$.

Thus, by Theorem \ref{theorem:2transitive_capacity}, a sequence of affine-invariant codes of increasing length, rates converging to $r \in (0,1)$, achieve capacity on the BEC under bit-MAP decoding.
Some examples of great interest include generalized Reed-Muller codes \cite[Corollary 2.5.3]{Delsarte-ic70} and extended primitive narrow-sense BCH codes \cite[Theorem 5.1.9]{Huffman-2003}.
Below, we discuss Reed-Muller and BCH codes in more detail.

\subsection{Reed-Muller Codes}
\label{subsection:reed_muller_codes}

For integers $v,m$ satisfying $0\leq v \leq m$, the Reed-Muller code $\mr{RM}(v,m)$ is a binary linear code with length $N=2^m$ and rate $r = 2^{-m} \left(\binom{m}{0}+\dots+\binom{m}{v}\right)$.
Although it is possible to describe these codes from the perspective of affine-invariance \cite[Corollary 2.5.3]{Delsarte-ic70}, below, we treat them as polynomial codes \cite{Kasami-it68*2}.
This provides a far more powerful insight to their structure \cite{Delsarte-ic70,Delsarte-it70}.

Consider the set of $m$ variables, $x_1,\dots,x_m$.
For a monomial $x_1^{i_1}\cdots x_m^{i_m}$ in these variables, define its degree to be $i_1 + \cdots + i_m$.
A polynomial in $m$ variables is the linear combination (using coefficients from a field) of such monomials and the degree of a polynomial is defined to be the maximum degree of any monomial it contains.
It is well-known that the set of all $m$-variable polynomials of degree at most $v$ is a vector space over its field of coefficients.
In this section, the coefficient field is the Galois field $\mbb{F}_2$ and the vector space of interest is given by
\begin{align*}
  P(m,v) \!=\! \text{span} \{ x_1^{t_1}\dots x_m^{t_m} \mid t_1\!+\dots+\!t_m \leq v, t_i \in \{0,1\}\} .
\end{align*}
For a polynomial $f \in P(m,v)$, $f(\ul{x}) \in \{0,1\}$ denotes the evaluation of $f$ at $\ul{x}\in \{0,1\}^m$. 

Let the elements of the vector space $\{0,1\}^m$ over $\mbb{F}_2$ be enumerated by $\ul{e}_1,\ul{e}_2,\dots,\ul{e}_{N}$ with $\ul{e}_N=0^m$.
For any polynomial $f \in P(m,v)$, we can evaluate $f$ at $\ul{e}_i$ for all $i\in [N]$.
Then, the code $\mr{RM}(v,m)$ is defined to be the set 
\begin{align*}
  \mr{RM}(v,m) \triangleq \{ (f(\ul{e}_1),\dots,f(\ul{e}_{N})) \mid f \in P(m,v) \} .
\end{align*}

\begin{lemma} [{\hspace{-0.01cm}\cite[Corollary 4]{Kasami-it68}}]
\label{lemma:rm_2transitive}
The permutation group $\mc{G}$ of $\mr{RM}(v,m)$ is doubly transitive.
\end{lemma}
\begin{IEEEproof}
Take any distinct $i,j,k \in [N]$.
Below, we will produce a $\pi \in \mc{G}$ such that $\pi(i)=i$ and $\pi(j)=k$.

It is well known that for any vector space with two ordered bases $(\ul{u}_1,\dots,\ul{u}_m)$ and $(\ul{u}'_1,\dots,\ul{u}'_m)$, there exists an invertible $m\times m$ matrix $T$ such that
\begin{align*}
  \ul{u}_i = T \ul{u}'_i , \quad \text{for all $i \in [m]$} .
\end{align*}

Note that since $i,j,k$ are distinct, $\ul{e}_j-\ul{e}_i \neq 0^m$ and $\ul{e}_k-\ul{e}_i \neq 0^m$.
Therefore, there exists an invertible $m \times m$ binary matrix $T$ such that
\begin{align*}
  T(\ul{e}_j-\ul{e}_i) = \ul{e}_k - \ul{e}_i .
\end{align*}
For such a $T$, we construct $\pi \colon [N] \to [N]$ by defining $\pi(\ell) = \ell'$ for the unique $\ell'$ such that 
\begin{align*}
  \ul{e}_{\ell'}= T (\ul{e}_{\ell} - \ul{e}_i) + e_i .
\end{align*}
Note that $\pi \in S_N$ since $T$ is invertible.
Also, by construction, $\pi(i)=i$ and $\pi(j)=k$.

It remains to show that $\pi \in \mc{G}$.
For this, consider a codeword in $\mr{RM}(v,m)$ given by $f \in P(m,v)$.
It suffices to produce a $g \in P(m,v)$ such that $g(\ul{e}_{\pi(\ell)})=f(\ul{e}_{\ell})$ for all $\ell \in [N]$.
Let
\begin{align*}
  g(x_1,\dots,x_m)=f(T^{-1}[x_1,\dots,x_m]^{\mr{T}} - T^{-1}\ul{e}_i + \ul{e}_i) , 
\end{align*}
and note that $\text{degree}(f)=\text{degree}(g)$, $g(\ul{e}_{\pi(\ell)})=f(\ul{e}_{\ell})$.
Thus, we have the desired $g \in P(m,v)$.
Hence, $\mc{G}$ is doubly transitive.
\end{IEEEproof}

There is also a sequence of $\{\mr{RM}(v_m,m)\}$ codes with increasing blocklengths and rates approaching any $r\in(0,1)$.
To construct such a sequence, fix $r\in(0,1)$ and let $\{ Z_i \}$ be an iid sequence of Bernoulli($1/2$) random variables.
Then, the rate of the $\mr{RM}(v_m,m)$ code is
\begin{align*}
  r_m &= \frac{1}{2^m} \left(\binom{m}{0}+\dots+\binom{m}{v_m}\right) \\
  &= \Pr(Z_1+\dots+Z_m \leq v_m)
   \\
  &= \Pr\left( \frac{Z_1-\frac{1}{2} + \dots + Z_m-\frac{1}{2}}{\sqrt{m/4}} \leq \frac{v_m-\tfrac{m}{2}}{\sqrt{m/4}} \right) .
\end{align*}
Thus, by central limit theorem, if we choose 
\begin{align*}
  v_{m}=\max \left\{ \left\lfloor \frac{m}{2} + \frac{\sqrt{m}}{2} Q^{-1}(1-r) \right\rfloor , 0 \right\},
\end{align*}
then the rate of $\mr{RM}(v_m,m)$ satisfies $r_m \to r$ as $m \to \infty$.
Here, 
\begin{align*}
  Q(t) \triangleq \tfrac{1}{\sqrt{2\pi}} \int_t^\infty e^{-\tau^2/2} \diff{\tau}.  
\end{align*}

\begin{theorem}
\label{theorem:rm_capacity_bitmap}
For any $r\in(0,1)$, the sequence of codes $\{ \mr{RM}(v_m,m) \}$ with
\begin{align*}
  v_{m} = \max \left\{ \left\lfloor \frac{m}{2} + \frac{\sqrt{m}}{2} Q^{-1}(1-r) \right\rfloor, 0 \right\},
\end{align*}
has rate $r_m \to r$ and is capacity achieving on the BEC under bit-MAP decoding.
\end{theorem}
\begin{IEEEproof}
This result follows as an immediate consequence of Lemma \ref{lemma:rm_2transitive} and Theorem \ref{theorem:2transitive_capacity}.
\end{IEEEproof}

We now analyze the block erasure probability of Reed-Muller codes.
The minimum distance of Reed-Muller codes is too small to utilize Theorem~\ref{theorem:capacity_blockerror_dmin}.
Thus, we use Theorem~\ref{theorem:capacity_blockerror_wn} instead.

For the code $\mr{RM}(v,m)$, consider the set $\Omega'_N$ from Definition \ref{definition:Omega} and \eqref{equation:Omega_prime}.
Let $\mc{G}_N$ be the permutation group of $\Omega'_N$ defined by
\begin{align*}
  \mc{G}_N \triangleq \{ \pi \in S_{N-1} \mid  \text{$\pi(\ul{a}) \in \Omega'_N$ for all $\ul{a} \in \Omega'_N$} \}.
\end{align*}

\begin{lemma}
\label{lemma:gl_g1}
For the permutation group $\mc{G}_N$ defined above, there is a transitive subgroup isomorphic to $\mr{GL}(m,\mbb{F}_2)$, the general linear group of degree $m$ over the Galois field $\mbb{F}_2$.
\end{lemma}
\begin{IEEEproof}
For a given $T \in \mr{GL}(m,\mbb{F}_2)$, associate $\pi_T \in S_{N-1}$, where 
\begin{align*}
  \pi_T(\ell) = \ell', \quad \text{where $\ul{e}_{\ell'} = T \ul{e}_{\ell}$}.
\end{align*}
Note that $\pi_T$ is well-defined since $T$ is invertible.
Moreover, it is easy to check that $\pi_{T_1} \circ \pi_{T_2} = \pi_{T_1 T_2}$ for $T_1,T_2 \in \mr{GL}(m,\mbb{F}_2)$.
As such, the collection of permutations
\begin{align*}
  \mc{H} = \{ \pi_T \in S_{N-1} \mid T \in \mr{GL}(m,\mbb{F}_2)\}
\end{align*}
is a subgroup of $S_{N-1}$ isomorphic to $\mr{GL}(m,\mbb{F}_2)$.
Also, for $i,j \in [N-1]$, there exists $T \in \mr{GL}(m,\mbb{F}_2)$ such that $\ul{e}_{j}=T\ul{e}_{i}$.
For such a $T$, $\pi_T(i)=j$.
Therefore, $\mc{H}$ is transitive.

It remains to show that $\mc{H} \subseteq \mc{G}_N$.
For this, associate $\pi_T \in \mc{H}$ with $\pi'_T \in S_{N}$ where
\begin{align*}
  \pi'_T(\ell) &= \pi_T(\ell)\quad \text{for $\ell \in [N-1]$}, & \pi'_T(N)&=N.
\end{align*}
Also, it is easy to show that $\pi_T \in \mc{G}_1$ if $\pi'_T \in \mc{G}$, the permutation group of $\mr{RM}(v,m)$.
To see that $\pi'_T \in \mc{G}$, consider a codeword given by $f \in P(m,v)$.
It suffices to produce a $g \in P(m,v)$ where $g(\ul{e}_{\pi'_T(\ell)})=f(\ul{e}_\ell)$ for $\ell \in [N]$.
The desired $g$ is given by $g(x_1,\dots,x_m)=f(T^{-1}[x_1,\dots,x_m]^{\mr{T}})$, by observing that $\text{degree}(g)=\text{degree}(f)$ and $g(\ul{e}_N)=f(T^{-1}0^m)=f(\ul{e}_N)$.
\end{IEEEproof}

\begin{theorem}
\label{theorem:rm_capacity_blockmap}
For any $r \in (0,1)$, the sequence of codes $\{ \mr{RM}(v_m,m) \}$, with
\begin{align*}
  v_{m} = \max \left\{ \left\lfloor \frac{m}{2} + \frac{\sqrt{m}}{2} Q^{-1}(1-r) \right\rfloor, 0 \right\},
\end{align*}
has rate $r_m \to r$ and is capacity achieving on the BEC under block-MAP decoding.
\end{theorem}
\begin{IEEEproof}
Let the EXIT function associated with the last bit and the average EXIT function of the code $\mr{RM}(v_m,m)$ be $h_N$ and $h$, respectively.
Since the permutation group of $\mr{RM}(v_m,m)$ is transitive by Lemma \ref{lemma:rm_2transitive}, from Proposition \ref{proposition:transitive_exit}, $h=h_N$.
Moreover, by Lemma \ref{lemma:gl_g1}, $\mc{G}_N$ contains a transitive subgroup isomorphic to $\mr{GL}(m,\mbb{F}_2)$.

Now, we can exploit the $\mr{GL}(m,\mbb{F}_2)$ symmetry of $\Omega_N$ within the framework of \cite{Bourgain-gafa97}.
In particular, \cite[Theorem 1, Corollary 4.1]{Bourgain-gafa97} implies that there exists a universal constant $C>0$, independent of $m$ and $p$, such that 
\begin{align*}
  \deri{h_N(p)}{p} \!\geq\! C \log (\log N_m) \log (N_m) h_N(p)(1-h_N(p)),
\end{align*}
for $0 < a_m < p < b_m < 1$, where $N_m = 2^m $ and $a_m \to 0$, $b_m \to 1$ as $m \to \infty$.
Since $h=h_N$, Theorem~\ref{theorem:capacity_blockerror_wn} implies that $\{\mr{RM}(v_m,m)\}$ is capacity achieving on the BEC under block-MAP decoding.
\end{IEEEproof}
From this, we see that the block erasure probability goes to $0$ for $p<1-r$. 
For $p>1-r$, the average EXIT function $h(p)$ is bounded away from $0$.
Thus, Theorem~\ref{theorem:rm_capacity_bitmap} implies that the bit erasure probability $ph(p)$ is bounded away from $0$ but not converging to $1$.
The block erasure probability does converge to $1$, however.  This follows from the result in \cite{Tillich-cpc00} because the minimum distance of the code $\mr{RM}(v_m,m)$ goes to $\infty$ as $m\to \infty$.

\subsection{Bose-Chaudhuri-Hocquengham Codes}
\label{subsection:bch_codes}

Let $\alpha$ be a primitive element of $\mbb{F}_{2^m}$.
Recall that a binary BCH code is \emph{primitive} if its blocklength is of the form $2^m-1$, and \emph{narrow-sense} if the roots of its generator polynomial include consecutive powers of a primitive element starting from $\alpha$.
In this article, we consider only primitive narrow-sense BCH codes and we follow closely the treatment of BCH codes in \cite{Huffman-2003}.

For integers $v$, $m$ with $1 \leq v \leq 2^m-1$, let $f(m,v)$ be the polynomial of lowest-degree over $\mbb{F}_2$ that has the roots
\begin{align*}
  \alpha,\alpha^2,\dots,\alpha^{v} .
\end{align*}
Then, $\mr{BCH}(v,m)$ is defined to be the binary cyclic code with the generator polynomial $f(m,v)$.
This is precisely the primitive narrow-sense BCH code with length $2^m -1$ and designed distance $v+1$.

The dimension $K$ of the cyclic code is determined by the degree of the generator polynomial according to \cite[Theorem 4.2.1]{Huffman-2003}
\begin{align*}
  K=N-\text{degree}(f(m,v)) .
\end{align*}
Moreover, the minimum distance $\dmin$ of $\mr{BCH}(v,m)$ is at least $v+1$ \cite[Theorem 5.1.1]{Huffman-2003}.

Since $\mbb{F}_{2^m}$ is the splitting field of the polynomial $x^N-1$ \cite[Theorem 3.3.2]{Huffman-2003}, it is easy to see that $\text{degree}(f(m,2^m-1)) = N$.
Also, since the size of the cyclotomic coset of any element $\alpha^i$ is at most $m$ \cite[Section 3.7]{Huffman-2003}, we have $\text{degree}(f(m,1)) \leq m$,
\begin{align*}
  0 \leq \text{degree}(f(m,v+1))-\text{degree}(f(m,v)) \leq m.
\end{align*}
Thus, for any $r \in(0,1)$, one can choose $v_{m}$ such that
\begin{align*}
  N(1-r) \leq \text{degree}(f(m,v_{m})) \leq N(1-r) + m .
\end{align*}
Now, it is easy to see that $v_{m} \geq N(1-r)/m$ and the rate of the code $\mr{BCH}(v_m,m)$ will be in $[r-\frac{m}{N},r]$.

Now, consider the length-$2^m$ extended BCH code, $\mr{eBCH}(v,m)$, which is formed by adding a single parity bit to the code $\mr{BCH}(v,m)$ so that overall codeword parity is always even~\cite[Section 5.1]{Huffman-2003}.
The code $\mr{eBCH}(v,m)$ has the same dimension as $\mr{BCH}(v,m)$ and a minimum distance of at least $v+1$.

Thus, for any $r\in(0,1)$, there exists a sequence of codes $\{ \mr{eBCH}(v_m,m) \}$ with blocklengths $N_m=2^m$, rates $r_m \to r$ and minimum distances
\begin{align}
  \label{equation:bch_dmin}
  \dmin^{(m)} \geq 1+v_{m} \geq 1+\frac{N_m(1-r)}{m} .
\end{align}

An important property of the extended BCH codes is that they are affine-invariant \cite[Theorem 5.1.9]{Huffman-2003}.
Thus, Section~\ref{subsection:affine_invariant_codes} shows that their permutation group is doubly transitive.
Therefore, we have the following theorem.
\begin{theorem}
\label{theorem:bch_capacity_bitmap}
For any $r \in (0,1)$, there is a sequence $\{v_m\}$ such that the code sequence $\{ \mr{eBCH}(v_m,m) \}$ has  $r_m\to r$ and is capacity achieving on the BEC under bit-MAP decoding.
\end{theorem}

In the following, we discuss the block erasure probability of BCH codes.
It is possible to characterize the permutation group of the code $\mr{eBCH}(v,m)$ precisely.
According to \cite{Berger-it96,Berger-des99}, except in sporadic cases, the permutation group of the code $\mr{eBCH}(v,m)$ is equal to the affine semi-linear group.
Unfortunately, in the framework of \cite{Bourgain-gafa97}, this group does not produce any factors beyond order $\log(N)$.
This is not encouraging for the analysis of block erasure probability.
This is in contrast with Reed-Muller codes where it was possible to exploit $\mr{GL}(m,\mbb{F}_2)$ symmetry to analyze their block erasure probability.
It is worth noting that the only cyclic primitive codes, whose permutation group includes the general linear group of degree $m$, are variants of generalized Reed-Muller codes~\cite{Delsarte-it70}.

For BCH codes, however, the minimum distance is large enough to use Theorem~\ref{theorem:capacity_blockerror_dmin}.
In fact, the minimum distance of the code $\mr{eBCH}(v_m,m)$ from \eqref{equation:bch_dmin} satisfies
\begin{align}
\label{equation:bch_dmin_limit}
  \lim_{m \to \infty} \frac{\log \dmin^{(m)}}{\log N_m} = 1.
\end{align}
Since the permutation group of the code $\mr{eBCH}(v_m,m)$ is doubly transitive from affine-invariance, by Theorem~\ref{theorem:width_bound} and the proof of Theorem \ref{theorem:2transitive_capacity}, its average EXIT function satisfies the hypothesis of Theorem~\ref{theorem:capacity_blockerror_dmin}.
Combining this observation with \eqref{equation:bch_dmin_limit} gives the following result.
\begin{theorem}
\label{theorem:bch_capacity_blockmap}
For any $r\in (0,1)$, there is a sequence $\{v_m\}$ such that the code sequence $\{ \mr{eBCH}(v_m,m) \}$ has  $r_m\to r$ and is capacity achieving on the BEC under block-MAP decoding.
\end{theorem}

\begin{corollary} \label{corollary:bch_capacity}
For any $r\in (0,1)$, there is a sequence $\{v_m\}$ such that the code sequence $\{ \mr{BCH}(v_m,m) \}$ has  $r_m\to r$ and is capacity achieving on the BEC under both bit-MAP and block-MAP decoding.
\end{corollary}
\begin{proof}
The code $\mr{BCH}(m,v)$ can be constructed from the code $\mr{eBCH}(m,v)$ simply by puncturing (i.e., erasing) the overall parity bit.
This implies that their EXIT functions satisfy $h^{\mr{eBCH}} (p) \geq p h^{\mr{BCH}}(p)$.
From this, we see that
\begin{equation*} \label{equation:BCHexit}
	h^{\mr{BCH}} (p) \leq \frac{1}{p} h^{\mr{eBCH}}(p).
\end{equation*}
Since puncturing single bit has an asymptotically negligible effect on the rate, the statement of the corollary follows directly from Theorems~\ref{theorem:bch_capacity_bitmap} and~\ref{theorem:bch_capacity_blockmap}.
\end{proof}

\begin{remark}
Corollary~\ref{corollary:bch_capacity} shows that there are sequences of binary cyclic codes that achieve capacity on the BEC.
As far as the authors know, this is the first proof that such a sequence exists~\cite{Ordentlich-it12}.
\end{remark}


\section{Discussion}
\label{section:discussion}

\subsection{Comparison with the Work of Tillich and Z{\'e}mor}
\label{section:TillichZemor}

Our initial attempts to prove a sharp threshold for EXIT functions focused on analyzing~\eqref{equation:boundary} with $\Omega=\Omega_i $.
In particular, our aim was to generalize~\cite{Tillich-cpc00} to EXIT functions by finding a lower bound on $g_{\Omega_i} (\ul{x})$ that holds uniformly over the boundary
$$\partial \Omega_i \triangleq \{ \ul{x}\in\{0,1\}^{N} \,|\, g_{\Omega_i} (\ul{x}) >0 \}.$$
For code sequences where $\dmin \to \infty$ and the minimum distance of the dual code satisfies $\dmin^{\perp} \to \infty$, we expected that $\min_{\ul{x}\in \partial \Omega_i} g_{\Omega_i} (\ul{x})$ would grow without bound and, thus, that the EXIT function would have a sharp threshold.
Unfortunately, this is not true.
In fact, the ensemble of $(j,k)$-regular LDPC codes provides a counterexample.
With high probability, their minimum distance grows linearly with $N$ but one iteration of iterative decoding shows that the EXIT function is upper bounded by $(1-(1-p)^{k-1})^j$ for all $p$ and $N$~\cite{RU-2008}. 

To understand this, first recall that a weight-$d$ codeword in the dual code defines a subset of $d$ code bits that sum to 0.
If only one of the bits in this dual codeword is erased, then that bit can be recovered indirectly from the other bits.
To see this in terms of the boundary, consider the indirect recovery of bit-$i$ and assume that it is contained in a weight-$d$ dual codeword with $d=\dmin^{\perp}\geq 3$.
Let $\ul{x}$ be an erasure pattern where $d-2$ of the $d-1$ other bits in the dual codeword are received correctly and all other bits are erased.
Then, $\ul{x} \in \Omega_i$ and bit-$i$ cannot be recovered indirectly.
Also, bit-$i$ can be recovered indirectly if the erased bit (say bit $j$) in the dual codeword is revealed.
Thus, $\ul{x}^{(j)} \notin \Omega_i$.

Now, let us consider $g_{\Omega_i} (\ul{x})$.
If there is any other bit (say bit $k$) for which $\ul{x}^{(k)} \notin \Omega_i$, then the pattern of correctly received symbols in $\ul{x}^{(k)}$ (along with bit $i$) must cover a dual codeword.
Since $\ul{x}^{(k)}$ contains exactly $d-1$ zero (i.e., unerased) symbols and the minimum dual distance is $d$, it follows that $\ul{x}^{(k)}$ must be a dual codeword.
Due to linearity, one can add the two vectors to get $\ul{x}^{(j)} + \ul{x}^{(k)}$, which clearly has weight 2.
However, this contradicts the assumption that the minimum dual distance is $\dmin^{\perp} \geq 3$.
Thus, we find that only bit $j$ is pivotal for $\ul{x}$ and
$$\min_{\ul{x} \in \partial \Omega_i} g_{\Omega_i} (\ul{x})=1.$$

This shows that the method of~\cite{Tillich-cpc00} does not extend automatically to prove sharp thresholds for EXIT functions.
While it is possible that there is a simple modification that overcomes this issue, we did not find it.

\subsection{Conditions of Theorem~\ref{theorem:2transitive_capacity}}
\label{section:relaxing_2transitivity}
One natural question is whether or not the conditions of Theorem~\ref{theorem:2transitive_capacity} can be weakened.
If the permutation groups of the codes in the sequence are not transitive, then different bits may have different EXIT functions with phase transitions at different values of $p$ (e.g., if some of the bits are protected by a random code of one rate and other bits with a random code of a different rate).

Even if the permutation groups are transitive, things can still go wrong.
Consider any sequence of codes with transitive permutation groups and increasing length.
Let $\{\dmin^{(n)}\}$ be the sequence of minimum distances.
Then, symmetry implies that the erasure rate of bit-MAP decoding is lower bounded by $p^{\dmin^{(n)}}$ for a BEC($p$) (e.g., every code bit is covered by a codeword with weight $\dmin$).
Thus, the sequence does not achieve capacity if $\dmin^{(n)}$ has a uniform upper bound.
Based on duality, a similar argument holds if the sequence of minimum dual distances $\{\dmin^{\perp (n)}\}$ is upper bounded.
Thus, to achieve capacity, a necessary condition is that $\dmin^{(n)} \to \infty$ and $\dmin^{\perp (n)} \to \infty$.
Based on this observation, we make the following optimistic conjecture.
\begin{conjecture}
Let $\{\mc{C}_n\}$ be a sequence of binary linear codes where the blocklengths satisfy $N_n \to \infty$, the rates satisfy $r_n \to r$ for $r\in (0,1)$, and the permutation group of each code is transitive.
If the sequence of minimum distances satisfies $\dmin^{(n)} \to \infty$ and the sequence of minimum dual distances satisfies $\dmin^{\perp (n)} \to \infty$, then the sequence achieves capacity on the BEC under bit-MAP decoding.
\end{conjecture}

We call a code \emph{reducible} if it can be written as the direct product of irreducible component codes of shorter length.
If a code is reducible, then the minimum distance of each irreducible component is at least as large as the minimum distance of the overall code.
Likewise, if the permutation group of a reducible code is transitive, then permutation group of each irreducible component code must also be transitive.
Moreover, transitivity implies that the EXIT function of each bit must equal both the EXIT function of the overall code and the EXIT function of any irreducible component code.
Thus, the rate of the overall code and the rate of each irreducible component code must all be equal to the integral of their common EXIT function.
This implies that, if the overall code satisfies the necessary conditions of the conjecture, then each of its irreducible component codes must also satisfy the necessary conditions.
Thus, it is sufficient to resolve the conjecture for the case where there is a single irreducible component code.

\subsection{Beyond the Erasure Channel}
\label{section:general_channels}
Beyond the erasure channel, this work also has implications for the decoding of Reed-Muller codes transmitted over the binary symmetric channel.
In particular, the results of~\cite[Theorem~1.8]{Abbe-arxiv14} show that an error pattern can be corrected by~$\mr{RM}(m-(2t+2),m)$ whenever an erasure pattern with the same support can be corrected by~$\mr{RM}(m-(t+1),m)$.
Such error patterns can even be corrected efficiently~\cite{Saptharishi-arxiv15}.

Another interesting open question is whether or not one can extend this approach to binary-input memoryless symmetric channels via generalized EXIT (GEXIT) functions~\cite{Measson-it09}.
For this, some new ideas will certainly be required because the straightforward approach leads to the analysis of functions that are neither boolean nor monotonic.

It would also be very interesting to find boolean functions outside of coding theory where area theorems can be used to pinpoint sharp thresholds.

\subsection{$\mbb{F}_q$-Linear Codes over the $q$-ary Erasure Channel}
\label{section:Fqlinear}

While our exposition focuses on binary linear codes over the BEC, it is easy to extend all results to $\mbb{F}_q$-linear codes over the $q$-ary erasure channel.

First, the set $\Omega_i$ is redefined to be the set of erasure patterns that prevent indirect recovery of the symbol $X_i$.
Importantly, $\Omega_i$ is still a set of binary sequences (equivalently, set of subsets of $[N] \backslash \{i\}$), and \emph{not} a set of sequences over the alphabet $\{0,1,\dots,q-1\}$.
Note that, if indirect recovery is not possible, then the linearity of the code implies that the posterior marginal of symbol $i$ given the extrinsic observations is $\Pr (X_i=x|\ul{Y}_{\sim i}=\ul{y}_{\sim i})=1/q$.
Next, we rescale the logarithm in the entropy $\ent{\cdot}$ to base $q$ so that $H(X_i | \ul{Y}_{\sim i}=\ul{y}_{\sim i})=1$ when indirect recovery of $X_i$ is not possible.

Thus, the sharp threshold framework for monotone boolean functions can be applied without change.
With these straightforward modifications, the results in Sections~\ref{section:preliminaries}~and~\ref{section:sharp_threshold} hold true verbatim.

The concept of affine-invariance also extends naturally to $\mbb{F}_q$-linear codes of length $q^m$ over the Galois field $\mbb{F}_q$.
Similarly, affine-invariance implies that the permutation group is doubly transitive.
Thus, sequences of affine-invariant $\mbb{F}_q$-linear codes of increasing length, whose rates converge to $r \in (0,1)$, achieve capacity over the $q$-ary erasure channel under symbol-MAP decoding.
The results for the block-MAP decoder also extend without change.
Thus, one finds that Generalized Reed-Muller codes~\cite{Delsarte-ic70} and extended primitive narrow-sense BCH codes over $\mbb{F}_q$ achieve capacity on the $q$-ary erasure channel under block-MAP decoding.

\subsection{Rates Converging to Zero}
\label{section:rates_to_zero}
Consider a sequence of Reed-Muller codes $\{\mr{RM}(v_m,m)\}$ where the rate $r_m \to 0$ sufficiently fast.
A key result of \cite{Abbe-arxiv14} is that Reed-Muller codes are capacity achieving in this scenario.
That is, for any $\delta>0$,
\begin{align*}
  P_B^{(m)}(p_m) \to 0 \quad \text{for any $0 \leq p_m < 1-(1+\delta) r_m$} .
\end{align*}
Looking closely at~\cite[Corollary 5.1]{Abbe-arxiv14}, it appears that $r_m = O(N_m^{-\kappa})$ for some $\kappa>0$ is a necessary condition for this result, where the blocklength $N_m=2^m$.

Let's analyze the bit erasure probability using our method.
From the proof of Theorem \ref{theorem:capacity_blockerror_wn}, it is possible to deduce that $P_b^{(m)}(p_{\epsilon_m}) \to 0$ if we choose $\epsilon_m = o(1)$ such that $\log(1/\epsilon_m) = o(\log(N_m))$.

We can also obtain a lower bound on $p_{\epsilon_m}$.
From the proof of Proposition \ref{proposition:capacity_equivalence}, we gather that
\begin{align*}
  p_{\epsilon_m} &\geq 1 - \frac{r_m}{1-\epsilon_m} - \left( p_{1-\epsilon_m} - p_{\epsilon_m} \right) .
\end{align*}
From Theorem \ref{theorem:width_bound} and the proof of Theorem \ref{theorem:2transitive_capacity},
\begin{align*}
  p_{1-\epsilon_m} - p_{\epsilon_m} \leq \frac{2\log \frac{1}{\epsilon_m} }{\log (N_m-1)} ,
\end{align*}
which implies that
\begin{align*}
  p_{\epsilon_m} \geq 1 - \frac{r_m}{1-\epsilon_m} - \frac{2 \log \frac{1}{\epsilon_m} }{\log (N_m-1)} 
  = 1 - (1+\delta_m) r_m  ,
\end{align*}
where 
\begin{align*}
  \delta_m = \frac{\epsilon_m}{1-\epsilon_m} + \frac{2 \log \frac{1}{\epsilon_m} }{r_m \log (N_m-1)} .
\end{align*}
Therefore,
\begin{align*}
  P_b^{(m)}(p_m) \to 0 \quad \text{for any $0 \leq p_m < 1 -  (1+\delta_m) r_m$} ,
\end{align*}
for any $\epsilon_m=o(1)$ such that $\log(1/\epsilon_m) = o(\log(N_m))$.

In order to obtain a capacity achieving result under bit-MAP decoding, we require that $\delta_m \to 0$.
This can be guaranteed if $r_m \log(N_m) \to \infty$.
Under this condition, we can choose $\epsilon_m = 1/\log(r_m \log(N_m))$ so that
\begin{align*}
  \epsilon_m & \to 0, & \frac{\log\tfrac{1}{\epsilon_m}}{\log(N_m)} & \to 0, & \delta_m &\to 0 .
\end{align*}
Thus, under the condition $r_m \log(N_m) \to \infty$, the sequence $\mr{RM}(v_m,m)$ achieves capacity on the BEC under bit-MAP decoding.

For $r_m \to 0$, our results require $r_m \log(N_m) \to \infty$ while the results in~\cite[Corollary 5.1]{Abbe-arxiv14} require $r_m = O(N_m^{-\kappa})$ for some $\kappa>0$.
Thus, the results in the two papers apply to two asymptotic rate regimes that are non-overlapping.

\section{Conclusion}
\label{section:conclusion}
We show that a sequence of binary linear codes achieves capacity if its blocklengths are strictly increasing, its code rates converge to some $r\in(0,1)$, and the permutation group of each code is doubly transitive.
To do this, we use isoperimetric inequalities for monotone boolean functions to exploit the symmetry of the codes.
This approach was successful largely because the transition point of the limiting EXIT function for the capacity-achieving codes is known a priori due to the area theorem.
One remarkable aspect of this method is its simplicity.
In particular, this approach does not rely on the precise structure of the code.

The main result extends naturally to $\mathbb{F}_q$-linear codes transmitted over a $q$-ary erasure channel under symbol-MAP decoding.
The class of affine-invariant $\mathbb{F}_q$-linear codes also achieve capacity, since their permutation group is doubly transitive.
Our results also show that Generalized Reed-Muller codes and extended primitive narrow-sense BCH codes achieve capacity on the $q$-ary erasure channel under block-MAP decoding.

\section*{Acknowledgments}
The authors's interest in this problem was piqued by its listing as an open problem during the 2015 Simons Institute program on Information Theory.
Henry Pfister would also like to acknowledge interesting conversations on this topic with Marco Mondelli at the Simons Institute.

\appendices
\section{Proof of Proposition~\ref{proposition:capacity_equivalence}}
\label{appendix:proof_capacity_equivalence}

$S1 \Longleftrightarrow S2$: First, recall from \eqref{equation:biterasure_exit_relation} that $P_b(p)=ph(p)$.
From this, it follows that $S2 \Longrightarrow S1$.
Now, consider $S1 \Longrightarrow S2$.
The relation $P_b(p)=ph(p)$ together with $P_b^{(n)}(p) \to 0$ and $h^{(n)}(0)=0$ implies
\begin{align*}
  \lim_{n \to \infty} h^{(n)}(p) = 0 \quad \text{for $0 \leq p < 1-r$}.
\end{align*}

Now, we focus on the limit of $h^{(n)} (p)$ for $1-r < p \leq 1$.
Fix $q\in(1-r,1]$ and choose $n_0$ large enough so that, for all $n>n_0$, we have $r_n > r - \epsilon$ and $h^{(n)}(1-r - \epsilon) \leq \epsilon$.
Such an $n_0$ exists because $r_n \to r$ and $h^{(n)}(p) \to 0$ for $0 \leq p<1-r$.
Since the function $h^{(n)}$ is increasing for all $n$, the EXIT area theorem (i.e., Proposition \ref{proposition:scalar_exit_properties}(c)) implies that, for all $n>n_0$, we have
\begin{align*}
  \allowdisplaybreaks
  r - \epsilon < r_n &= \int_0^1 h^{(n)}(p) \diff{p} \\
  &= \int_0^{1-r-\epsilon} h^{(n)}(p)\diff{p} + \int_{1-r-\epsilon}^{q} h^{(n)}(p) \diff{p} \\
  & \qquad \qquad + \int_{q}^{1} h^{(n)}(p)\diff{p} \\
  &\leq (1-r-\epsilon) \epsilon + (q-(1-r)+\epsilon) h^{(n)}(q) \\
  & \qquad \qquad + (1-q).
\end{align*}
This implies
\begin{align*}
  h^{(n)}(q) & \geq \frac{q-(1-r) - \epsilon(2-r-\epsilon)}{q-(1-r)+\epsilon} \\
  &\geq 1 - \frac{2\epsilon}{q-(1-r)} .
\end{align*}
As such, $\lim_{n\to\infty} h^{(n)}(q) = 1$, for any $1-r<q\leq 1$.

$S2 \Longrightarrow S3$: Since $p_{1-\epsilon}^{(n)} - p_{\epsilon}^{(n)}$ is the width of the erasure probability interval over which $h^{(n)}$ transitions from $\epsilon$ to $1-\epsilon$, this follows immediately from $S2$.

$S3 \Longrightarrow S2$:
It suffices to show that, for any $\epsilon \in (0,1/2]$,
\begin{align*}
  \lim_{n \to \infty}p_{\epsilon}^{(n)} = \lim_{n \to \infty}p_{1-\epsilon}^{(n)} = 1-r.
\end{align*}
From Proposition \ref{proposition:scalar_exit_properties}(c), we have
\begin{align*}
  r_n &= \int_{0}^{1} h^{(n)}(\alpha) \diff{\alpha} \leq \epsilon p_\epsilon^{(n)} + \left(1 - p_\epsilon^{(n)} \right) ,
\end{align*}
which implies $p_\epsilon^{(n)} \leq \frac{1-r_n}{1-\epsilon}$.
Similarly,
\begin{align*}
  r_n &= \int_{0}^{1} h^{(n)}(\alpha) \diff{\alpha} \geq \left(1-p_{1-\epsilon}^{(n)} \right) (1-\epsilon),
\end{align*}
which implies
\begin{align*}
  p_{1-\epsilon}^{(n)} \geq \frac{1-r_n-\epsilon}{1-\epsilon} .
\end{align*}
Combining these gives
\begin{align*}
  \frac{1-r_n-\epsilon}{1-\epsilon} + \left( p_{\epsilon}^{(n)} - p_{1-\epsilon}^{(n)} \right)  \leq p_{\epsilon}^{(n)} \leq \frac{1-r_n}{1-\epsilon} , \\
  \frac{1-r_n-\epsilon}{1-\epsilon}  \leq p_{1-\epsilon}^{(n)} \leq \frac{1-r_n}{1-\epsilon} + \left( p_{1-\epsilon}^{(n)} - p_{\epsilon}^{(n)} \right).
\end{align*}
From the hypothesis, 
\begin{align*}
  \frac{1-r-\epsilon}{1-\epsilon} \leq \limsup_{n \to \infty} p_{\epsilon}^{(n)} \leq \frac{1-r}{1-\epsilon} ,\\
  \frac{1-r-\epsilon}{1-\epsilon} \leq \limsup_{n \to \infty} p_{1-\epsilon}^{(n)} \leq \frac{1-r}{1-\epsilon} .
\end{align*}
Thus
\begin{align*}
  \lim_{t \to 0} \limsup_{n\to\infty} p_{t}^{(n)} = \lim_{t \to 0} \limsup_{n\to\infty} p_{1-t}^{(n)} = 1-r.
\end{align*}
But $p_{t}^{(n)}$ and $p_{1-t}^{(n)}$ are increasing and decreasing functions of $t$, respectively.
This gives
\begin{align*}
  \limsup_{n\to\infty} p_{\epsilon}^{(n)} &\geq \lim_{t \to 0} \limsup_{n\to\infty} p_{t}^{(n)} = 1-r , \\
  \limsup_{n\to\infty} p_{1-\epsilon}^{(n)} &\leq \lim_{t \to 0} \limsup_{n\to\infty} p_{1-t}^{(n)} = 1-r .
\end{align*}
Since $p_{\epsilon}^{(n)} \leq p_{1-\epsilon}^{(n)}$, we deduce that
\begin{align*}
  \limsup_{n \to \infty} p_{\epsilon}^{(n)} = \limsup_{n \to \infty} p_{1-\epsilon}^{(n)} = 1-r.
\end{align*}
Repeating this exercise with $\limsup$ replaced by $\liminf$ also gives the result $1-r$.  Thus, for any $\epsilon \in (0,1/2]$, we have
\begin{align*}
  \lim_{n \to \infty} p_{\epsilon}^{(n)} = \lim_{n \to \infty} p_{1-\epsilon}^{(n)} = 1-r.
\end{align*}

\section{Proofs from Section~\ref{section:sharp_threshold}}
\label{appendix:proofs_section_inequalities}

\begin{lemma}
\label{lemma:exit_integral_generic}
Suppose $h \colon [0,1] \to [0,1]$ is a strictly increasing function with $h(0)=0$ and $h(1)=1$.
Additionally, for $0 \leq a < p < b \leq 1$, let
\begin{align*}
  \deri{h(p)}{p} \geq w h(p) (1-h(p)) .
\end{align*}
If $p_t = h^{-1}(t)$, then for $0 < \epsilon_1 \leq \epsilon_2 \leq 1$,
\begin{align}
  \label{equation:width_upper_bound}
  p_{\epsilon_2} \!-\! p_{\epsilon_1} \leq a + (1\!-\!b) + \frac{1}{w} \left[ \log\frac{\epsilon_2}{1\!-\!\epsilon_2} + \log\frac{1\!-\!\epsilon_1}{\epsilon_1} \right] .
\end{align}
Moreover, for $0 \leq \delta \leq p_{1/2}$,
\begin{align*}
  h(\delta) \leq \exp \left[ -w \left( [p_{1/2} - \delta] -[a + 1-b] \right) \right] .
\end{align*}
\end{lemma}
\begin{IEEEproof}
Let $g (p) = \log \frac{h(p)}{1-h(p)}$ and observe that, for $a < p < b$, we have
\begin{align*}
 \deri{g(p)}{p}=\frac{1}{h(p)(1-h(p))} \deri{h(p)}{p} \geq w .
\end{align*}
Let $p_t=h^{-1}(t)$.
We would like to obtain an upper bound on $p_{\epsilon_2} - p_{\epsilon_1}$ by integrating $dg/dp$.

If $a < p_{\epsilon_1} \leq p_{\epsilon_2} < b$, then integrating $dg/dp$ from $p_{\epsilon_1}$ to $p_{\epsilon_2}$ gives
\begin{align*}
  w (p_{\epsilon_2} \!-\! p_{\epsilon_1}) \leq \int_{p_{\epsilon_1}}^{p_{\epsilon_2}} \frac{dg}{dp} \diff{p} = \log \frac{\epsilon_2}{1\!-\!\epsilon_2} - \log \frac{\epsilon_1}{1\!-\!\epsilon_1} ,
\end{align*}
which immediately shows \eqref{equation:width_upper_bound}.

Suppose $p_{\epsilon_1} \leq a < p_{\epsilon_2} < b$, and note that since $g$ is increasing $\epsilon_1 = g(p_{\epsilon_1}) \leq g(a)$.
Then, integrating $dg_1/dp$ from $a$ to $p_{\epsilon_2}$ gives
\begin{align*}
  w (p_{\epsilon_2} - a) &\leq \int_{a}^{p_{\epsilon_2}} \frac{dg}{dp} \diff{p} \\
  &= \log \frac{\epsilon_2}{1-\epsilon_2} - \log \frac{h(a)}{1-h(a)} \\
  &\leq \log \frac{\epsilon_2}{1-\epsilon_2} - \log \frac{\epsilon_1}{1-\epsilon_1} . \quad \text{(Since $\epsilon_1 \leq h(a)$)}
\end{align*}
Using $p_{\epsilon_2} - p_{\epsilon_1} \leq a + (p_{\epsilon_2} - a)$ with the above inequality gives \eqref{equation:width_upper_bound}.

By considering other cases where $p_{\epsilon_1}$ and $p_{\epsilon_2}$ lie, it is straightforward to obtain \eqref{equation:width_upper_bound}.
Also, substituting $\epsilon_2=1/2$ and $\epsilon_1=h^{-1}(\delta)$ in \eqref{equation:width_upper_bound} gives the desired upper bound on $h(\delta)$.
\end{IEEEproof}

\subsection{Proof of Theorem~\ref{theorem:capacity_blockerror_dmin}}
\label{appendix:proof_capacity_blockerror_dmin}

Let $p^{(n)}_t$ be the functional inverse of $h^{(n)}$ from \eqref{equation:h_inverse}.
Using Lemma~\ref{lemma:exit_integral_generic} with $a_n=0$ and $b_n=1$ gives
\begin{align*}
  p^{(n)}_{1-\epsilon} - p^{(n)}_{\epsilon} \leq \frac{2 \log \frac{1- \epsilon}{\epsilon}}{C \log N_n} .
\end{align*}
By hypothesis, $N_n \to \infty$.
Thus, for any $\epsilon \in (0,1/2]$, we have $p^{(n)}_{1-\epsilon}-p^{(n)}_{\epsilon} \to 0$.
Using this, we apply statement $S2$ of Proposition~\ref{proposition:capacity_equivalence} to see that $p^{(n)}_{1/2}\to 1-r$.

Now, we can choose $\epsilon_n = \dmin^{(n)}/(N_n \log N_n)$ and observe that
\begin{align*}
  p^{(n)}_{1-\epsilon_n}- p^{(n)}_{\epsilon_n} &\leq \frac{2}{C \log N_n} \log \frac{1-\epsilon_n}{\epsilon_n} \\
  &\leq \frac{2}{C \log N_n} \log \frac{N_n \log N_n}{\dmin^{(n)}} \\
  &= \frac{2}{C} \frac{\log N_n + \log \log N_n - \log \dmin^{(n)}}{\log N_n} .
\end{align*}
By hypothesis, $\log \dmin^{(n)}/\log N_n \to 1$.
Thus, $p^{(n)}_{1-\epsilon_n}-p^{(n)}_{\epsilon_n} \to 0$.
Combining this with $p^{(n)}_{\epsilon_n} \leq p^{(n)}_{1/2} \leq p^{(n)}_{1-\epsilon_n}$ shows that $p^{(n)}_{\epsilon_n} \to 1-r$.

Also, from \eqref{equation:biterasure_exit_relation},
\begin{align*}
  P_b^{(n)}(p^{(n)}_{\epsilon_n}) = p^{(n)}_{\epsilon_n} h^{(n)}(p^{(n)}_{\epsilon_n}) \leq h^{(n)}(p^{(n)}_{\epsilon_n}) = \epsilon_n .
\end{align*}
Recall from \eqref{equation:erasure_prob_dmin} that $P_B \leq N P_b / \dmin$.
Hence, for any $p \in [0,1-r)$, one finds that $p^{(n)}_{\epsilon_n} > p$ for sufficiently large $n$ and thereafter
$$P_B^{(n)}(p) \leq \frac{N_n}{\dmin^{(n)}} P_b^{(n)}(p) \leq \frac{N_n}{\dmin^{(n)}} \epsilon_n = \frac{1}{\log N_n} \to 0.$$
Thus, we conclude that $\{\mc{C}_n\}$ is capacity achieving on the BEC under block-MAP decoding.

\subsection{Proof of Theorem~\ref{theorem:capacity_blockerror_wn}}
\label{appendix:proof_capacity_blockerror_wn}

Let $p^{(n)}_t$ be the functional inverse of $h^{(n)}$ from \eqref{equation:h_inverse}.
From Lemma~\ref{lemma:exit_integral_generic},
\begin{align*}
  p^{(n)}_{1-\epsilon} - p^{(n)}_{\epsilon} \leq a_n + (1-b_n) + \frac{2 \log \frac{1- \epsilon}{\epsilon}}{w_n \log N_n} .
\end{align*}
By hypothesis, $a_n \to 0$, $1-b_n \to 0$, and $w_n \log N_n \to \infty$.
Thus, for any $\epsilon \in (0,1/2]$, we have $p^{(n)}_{1-\epsilon}-p^{(n)}_{\epsilon} \to 0$.
Using this, we apply statement $S2$ of Proposition~\ref{proposition:capacity_equivalence} to see that $p^{(n)}_{1/2}\to 1-r$.

Now, we can choose $\epsilon_n = 1/N_n^2$ and observe that
\begin{align*}
  p^{(n)}_{1-\epsilon_n}- p^{(n)}_{\epsilon_n} &\leq a_n \!+\! (1-b_n) \!+\! \frac{1}{w_n\log N_n} 2 \log \frac{1-\epsilon_n}{\epsilon_n} \\
  &\leq a_n \!+\! (1-b_n) + \frac{1}{w_n \log N_n} 4 \log N_n \\
  &= a_n \!+\! (1-b_n) \!+\! \frac{4}{w_n}.
\end{align*}
Combining $p^{(n)}_{\epsilon_n} \leq p^{(n)}_{1/2} \leq p^{(n)}_{1-\epsilon_n}$ with $p^{(n)}_{1-\epsilon_n}-p^{(n)}_{\epsilon_n} \to 0$ shows that $p^{(n)}_{\epsilon_n} \to 1-r$.

Also, from \eqref{equation:biterasure_exit_relation},
\begin{align*}
  P_b^{(n)}(p^{(n)}_{\epsilon_n}) = p^{(n)}_{\epsilon_n} h^{(n)}(p^{(n)}_{\epsilon_n}) \leq h^{(n)}(p^{(n)}_{\epsilon_n}) = \epsilon_n .
\end{align*}
Recall from \eqref{equation:erasure_prob_relations} that $P_B \leq N P_b$.
Hence, for any $p \in [0,1-r)$, one finds that $p^{(n)}_{\epsilon_n} > p$ for sufficiently large $n$ and thereafter
$$P_B^{(n)}(p) \leq N_n P_b^{(n)}(p) \leq N_n \epsilon_n = N_n/N_n^2 \to 0.$$
Thus, we conclude that $\{\mc{C}_n\}$ is capacity achieving on the BEC under block-MAP decoding.


\bibliographystyle{IEEEtran}
\bibliography{WCLabrv,WCLbib,WCLnewbib}

\begin{thebibliography}{10}
\providecommand{\url}[1]{#1}
\csname url@samestyle\endcsname
\providecommand{\newblock}{\relax}
\providecommand{\bibinfo}[2]{#2}
\providecommand{\BIBentrySTDinterwordspacing}{\spaceskip=0pt\relax}
\providecommand{\BIBentryALTinterwordstretchfactor}{4}
\providecommand{\BIBentryALTinterwordspacing}{\spaceskip=\fontdimen2\font plus
\BIBentryALTinterwordstretchfactor\fontdimen3\font minus
  \fontdimen4\font\relax}
\providecommand{\BIBforeignlanguage}[2]{{%
\expandafter\ifx\csname l@#1\endcsname\relax
\typeout{** WARNING: IEEEtran.bst: No hyphenation pattern has been}%
\typeout{** loaded for the language `#1'. Using the pattern for}%
\typeout{** the default language instead.}%
\else
\language=\csname l@#1\endcsname
\fi
#2}}
\providecommand{\BIBdecl}{\relax}
\BIBdecl

\bibitem{Shannon-bell48}
C.~E. Shannon, ``A mathematical theory of communication,'' \emph{The Bell
  Syst.\ Techn.\ J.}, vol.~27, pp. 379--423, 623--656, July / Oct. 1948.

\bibitem{Berrou-icc93}
C.~Berrou, A.~Glavieux, and P.~Thitimajshima, ``Near {Shannon} limit
  error-correcting coding and decoding: Turbo-codes,'' in \emph{Proc.\ IEEE
  Int.\ Conf.\ Commun.}, vol.~2.\hskip 1em plus 0.5em minus 0.4em\relax Geneva,
  Switzerland: IEEE, May 1993, pp. 1064--1070.

\bibitem{Gallager-1963}
R.~G. Gallager, \emph{Low-Density Parity-Check Codes}.\hskip 1em plus 0.5em
  minus 0.4em\relax Cambridge, MA, USA: The M.I.T. Press, 1963.

\bibitem{Spielman-it96}
D.~Spielman, ``Linear-time encodable and decodable error-correcting codes,''
  \emph{IEEE Trans.\ Inform.\ Theory}, vol.~42, no.~6, pp. 1723--1731, Nov
  1996.

\bibitem{Mackay-it99}
D.~J.~C. MacKay, ``Good error-correcting codes based on very sparse matrices,''
  \emph{IEEE Trans.\ Inform.\ Theory}, vol.~45, no.~2, pp. 399--431, March
  1999.

\bibitem{Luby-it01}
M.~G. Luby, M.~Mitzenmacher, M.~A. Shokrollahi, and D.~A. Spielman, ``Efficient
  erasure correcting codes,'' \emph{IEEE Trans.\ Inform.\ Theory}, vol.~47,
  no.~2, pp. 569--584, Feb. 2001.

\bibitem{Arikan-it09}
E.~Ar{\i}kan, ``Channel polarization: {A} method for constructing
  capacity-achieving codes for symmetric binary-input memoryless channels,''
  \emph{IEEE Trans.\ Inform.\ Theory}, vol.~55, no.~7, pp. 3051--3073, July
  2009.

\bibitem{Kudekar-it11}
S.~Kudekar, T.~J. Richardson, and R.~L. Urbanke, ``Threshold saturation via
  spatial coupling: {W}hy convolutional {LDPC} ensembles perform so well over
  the {BEC},'' \emph{IEEE Trans.\ Inform.\ Theory}, vol.~57, no.~2, pp.
  803--834, Feb. 2011.

\bibitem{Lentmaier-it10}
M.~Lentmaier, A.~Sridharan, D.~J. Costello, and K.~S. Zigangirov, ``Iterative
  decoding threshold analysis for {LDPC} convolutional codes,'' \emph{IEEE
  Trans.\ Inform.\ Theory}, vol.~56, no.~10, pp. 5274--5289, Oct. 2010.

\bibitem{Kudekar-it13}
S.~Kudekar, T.~Richardson, and R.~Urbanke, ``Spatially coupled ensembles
  universally achieve capacity under belief propagation,'' \emph{IEEE Trans.\
  Inform.\ Theory}, vol.~59, no.~12, pp. 7761--7813, Dec. 2013.

\bibitem{Kumar-it14}
S.~Kumar, A.~J. Young, N.~Macris, and H.~D. Pfister, ``Threshold saturation for
  spatially-coupled {LDPC} and {LDGM} codes on {BMS} channels,'' \emph{IEEE
  Trans.\ Inform.\ Theory}, vol.~60, no.~12, pp. 7389--7415, Dec. 2014.

\bibitem{Delsarte-ic70}
P.~Delsarte, J.~Goethals, and F.~M. Williams, ``On generalized {R}eed-{M}uller
  codes and their relatives,'' \emph{Inform.\ and Control}, vol.~16, no.~5, pp.
  403--442, 1970.

\bibitem{Kasami-it68}
T.~Kasami, S.~Lin, and W.~Peterson, ``New generalizations of the
  {R}eed-{M}uller codes--{I}: Primitive codes,'' \emph{IEEE Trans.\ Inform.\
  Theory}, vol.~14, no.~2, pp. 189--199, Mar 1968.

\bibitem{Kasami-ic68}
T.~Kasami, S.~Lin, and W.~W. Peterson, ``Some results on cyclic codes which are
  invariant under the affine group and their applications,'' \emph{Inform.\ and
  Control}, vol.~11, no.~5, pp. 475--496, 1968.

\bibitem{Ahlswede-it82}
R.~Ahlswede and G.~Dueck, ``Good codes can be produced by a few permutations,''
  \emph{IEEE Trans.\ Inform.\ Theory}, vol.~28, no.~3, pp. 430--443, May 1982.

\bibitem{Coffrey-it90}
J.~Coffey and R.~Goodman, ``Any code of which we cannot think is good,''
  \emph{IEEE Trans.\ Inform.\ Theory}, vol.~36, no.~6, pp. 1453--1461, Nov
  1990.

\bibitem{Muller-ire54}
D.~Muller, ``Application of boolean algebra to switching circuit design and to
  error detection,'' \emph{IRE Tran. on Electronic Computers}, vol. EC-3,
  no.~3, pp. 6--12, Sept 1954.

\bibitem{Reed-ire54}
I.~Reed, ``A class of multiple-error-correcting codes and the decoding
  scheme,'' \emph{IRE Tran. on Information Theory}, vol.~4, no.~4, pp. 38--49,
  September 1954.

\bibitem{Macwilliams-1977}
F.~J. MacWilliams and N.~J.~A. Sloane, \emph{The theory of error correcting
  codes}.\hskip 1em plus 0.5em minus 0.4em\relax Elsevier, 1977, vol.~16.

\bibitem{Lin-2004}
S.~Lin and D.~J. Costello, Jr., \emph{Error Control Coding: Fundamentals and
  Applications}, 2nd~ed.\hskip 1em plus 0.5em minus 0.4em\relax Englewood
  Cliffs, NJ, USA: Prentice-Hall, 2004, iSBN-13: 978-0130426727.

\bibitem{Costello-proc07}
D.~J. Costello, Jr. and G.~D. Forney, Jr., ``Channel coding: The road to
  channel capacity,'' \emph{Proc.\ of the IEEE}, vol.~95, no.~6, pp.
  1150--1177, June 2007.

\bibitem{Didier-it06}
F.~Didier, ``A new upper bound on the block error probability after decoding
  over the erasure channel,'' \emph{IEEE Trans.\ Inform.\ Theory}, vol.~52,
  no.~10, pp. 4496--4503, Oct 2006.

\bibitem{Mondelli-tcom14}
M.~Mondelli, S.~Hassani, and R.~Urbanke, ``From polar to {R}eed-{M}uller codes:
  A technique to improve the finite-length performance,'' \emph{IEEE Trans.\
  Commun.}, vol.~62, no.~9, pp. 3084--3091, Sept 2014.

\bibitem{Arikan-itw10}
E.~Arikan, ``A survey of reed-muller codes from polar coding perspective,'' in
  \emph{Proc.\ IEEE Inform.\ Theory Workshop}, Jan 2010, pp. 1--5.

\bibitem{Abbe-arxiv14}
E.~Abbe, E.~Shpilka, and A.~Wigderson, ``{R}eed-{M}uller codes for random
  erasures and errors,'' to appear in STOC 15, arXiv:1411.4590.

\bibitem{Carlet-isit05}
C.~Carlet and P.~Gaborit, ``On the construction of balanced boolean functions
  with a good algebraic immunity,'' in \emph{Proc.\ IEEE Int.\ Symp.\ Inform.\
  Theory}, Sept 2005, pp. 1101--1105.

\bibitem{Sloane-it70}
N.~Sloane and E.~Berlekamp, ``Weight enumerator for second-order
  {R}eed-{M}uller codes,'' \emph{IEEE Trans.\ Inform.\ Theory}, vol.~16, no.~6,
  pp. 745--751, Nov 1970.

\bibitem{Kasami-it70}
T.~Kasami and N.~Tokura, ``On the weight structure of {R}eed-{M}uller codes,''
  \emph{IEEE Trans.\ Inform.\ Theory}, vol.~16, no.~6, pp. 752--759, Nov 1970.

\bibitem{Kasami-ic76}
T.~Kasami, N.~Tokura, and S.~Azumi, ``On the weight enumeration of weights less
  than 2.5d of {R}eed-{M}uller codes,'' \emph{Inform.\ and Control}, vol.~30,
  no.~4, pp. 380 -- 395, 1976.

\bibitem{Kaufman-it12}
T.~Kaufman, S.~Lovett, and E.~Porat, ``Weight distribution and list-decoding
  size of {R}eed-{M}uller codes,'' \emph{IEEE Trans.\ Inform.\ Theory},
  vol.~58, no.~5, pp. 2689--2696, May 2012.

\bibitem{Sidelnikov-ppi92}
V.~M. Sidel'nikov and A.~Pershakov, ``Decoding of {R}eed-{M}uller codes with a
  large number of errors,'' \emph{Problems of Inform.\ Transm.}, vol.~28,
  no.~3, pp. 80--94, 1992.

\bibitem{Dumer-it04}
I.~Dumer, ``Recursive decoding and its performance for low-rate {R}eed-{M}uller
  codes,'' \emph{IEEE Trans.\ Inform.\ Theory}, vol.~50, no.~5, pp. 811--823,
  May 2004.

\bibitem{Dumer-it06*2}
------, ``Soft-decision decoding of reed-muller codes: a simplified
  algorithm,'' \emph{IEEE Trans.\ Inform.\ Theory}, vol.~52, no.~3, pp.
  954--963, March 2006.

\bibitem{Dumer-it06}
I.~Dumer and K.~Shabunov, ``Soft-decision decoding of {R}eed-{M}uller codes:
  recursive lists,'' \emph{IEEE Trans.\ Inform.\ Theory}, vol.~52, no.~3, pp.
  1260--1266, March 2006.

\bibitem{Saptharishi-arxiv15}
R.~{Saptharishi}, A.~{Shpilka}, and B.~L. {Volk}, ``Decoding high rate
  {R}eed-{M}uller codes from random errors in near linear time,'' 2015,
  [Online]. Available: http://arxiv.org/abs/1503.09092.

\bibitem{Arikan-comlett08}
E.~Arikan, ``A performance comparison of polar codes and {R}eed-{M}uller
  codes,'' \emph{IEEE Commun.\ Letters}, vol.~12, no.~6, pp. 447--449, June
  2008.

\bibitem{Camion-acrypt92}
P.~Camion, C.~Carlet, P.~Charpin, and N.~Sendrier, ``On correlation-immune
  functions,'' in \emph{Advances in Cryptology--CRYPTO’91}.\hskip 1em plus
  0.5em minus 0.4em\relax Springer, 1992, pp. 86--100.

\bibitem{Ta-shma-focs01}
A.~Ta-Shma, D.~Zuckerman, and S.~Safra, ``Extractors from {R}eed-{M}uller
  codes,'' in \emph{Proc.\ IEEE Symp.\ on the Found.\ of Comp.\ Sci.}\hskip 1em
  plus 0.5em minus 0.4em\relax IEEE, 2001, pp. 638--647.

\bibitem{Shaltiel-focs01}
R.~Shaltiel and C.~Umans, ``Simple extractors for all min-entropies and a new
  pseudo-random generator,'' in \emph{Proc.\ IEEE Symp.\ on the Found.\ of
  Comp.\ Sci.}\hskip 1em plus 0.5em minus 0.4em\relax IEEE, 2001, pp. 648--657.

\bibitem{Canteaut-it01}
A.~Canteaut, C.~Carlet, P.~Charpin, and C.~Fontaine, ``On cryptographic
  properties of the cosets of {R(1, m)},'' \emph{IEEE Trans.\ Inform.\ Theory},
  vol.~47, no.~4, pp. 1494--1513, 2001.

\bibitem{Carlet-it06}
C.~Carlet, D.~K. Dalai, K.~C. Gupta, and S.~Maitra, ``Algebraic immunity for
  cryptographically significant boolean functions: analysis and construction,''
  \emph{IEEE Trans.\ Inform.\ Theory}, vol.~52, no.~7, pp. 3105--3121, 2006.

\bibitem{Didier-fse06}
F.~Didier and J.-P. Tillich, ``Computing the algebraic immunity efficiently,''
  in \emph{Fast Software Encryption}.\hskip 1em plus 0.5em minus 0.4em\relax
  Springer, 2006, pp. 359--374.

\bibitem{Gerard-alc11}
B.~G{\'e}rard and J.-P. Tillich, ``Using tools from error correcting theory in
  linear cryptanalysis,'' \emph{Adv. Linear Cryptanalysis of Block and Stream
  Ciphers}, vol.~7, p.~87, 2011.

\bibitem{Yekhanin-fnt11}
S.~Yekhanin, ``Locally decodable codes,'' \emph{Found. Trends Theor. Comput.
  Sci.}, vol.~7, no.~4, pp. 169--174, 1992.

\bibitem{Gemmell-stoc91}
P.~Gemmell, R.~Lipton, R.~Rubinfeld, M.~Sudan, and A.~Wigderson,
  ``Self-testing/correcting for polynomials and for approximate functions,'' in
  \emph{STOC}, vol.~91.\hskip 1em plus 0.5em minus 0.4em\relax Citeseer, 1991,
  pp. 32--42.

\bibitem{Gemmell-ipl92}
P.~Gemmell and M.~Sudan, ``Highly resilient correctors for polynomials,''
  \emph{Information processing letters}, vol.~43, no.~4, pp. 169--174, 1992.

\bibitem{Kaufman-approxrandom10}
T.~Kaufman and M.~Viderman, ``Locally testable vs. locally decodable codes,''
  in \emph{Approximation, Randomization, and Combinatorial Optimization.
  Algorithms and Techniques}.\hskip 1em plus 0.5em minus 0.4em\relax Springer,
  2010, pp. 670--682.

\bibitem{Grigorescu-ccc08}
E.~Grigorescu, T.~Kaufman, and M.~Sudan, ``2-transitivity is insufficient for
  local testability,'' in \emph{Annual IEEE Conf. on Comp. Complex.}, June
  2008, pp. 259--267.

\bibitem{tenBrink-elet99}
S.~ten Brink, ``Convergence of iterative decoding,'' \emph{Electronic Letters},
  vol.~35, no.~10, pp. 806--808, May 1999.

\bibitem{Ashikhmin-it04}
A.~Ashikhmin, G.~Kramer, and S.~{ten Brink}, ``Extrinsic information transfer
  functions: model and erasure channel properties,'' \emph{IEEE Trans.\
  Inform.\ Theory}, vol.~50, no.~11, pp. 2657--2674, Nov. 2004.

\bibitem{Measson-it08}
C.~M{\'e}asson, A.~Montanari, and R.~L. Urbanke, ``Maxwell construction: The
  hidden bridge between iterative and maximum a posteriori decoding,''
  \emph{IEEE Trans.\ Inform.\ Theory}, vol.~54, no.~12, pp. 5277--5307, Dec.
  2008.

\bibitem{RU-2008}
T.~J. Richardson and R.~L. Urbanke, \emph{Modern Coding Theory}.\hskip 1em plus
  0.5em minus 0.4em\relax New York, NY: Cambridge University Press, 2008.

\bibitem{Boucheron-2013}
S.~Boucheron, G.~Lugosi, and P.~Massart, \emph{Concentration inequalities: A
  nonasymptotic theory of independence}.\hskip 1em plus 0.5em minus 0.4em\relax
  Oxford University Press, 2013.

\bibitem{Kalai-05}
G.~Kalai and S.~Safra, ``Threshold phenomena and influence with some
  perspectives from mathematics, computer science, and economics,'' \emph{Comp.
  Complexity and Stat. Phy., Santa Fe Institute Studies in Sci. of Complexity},
  vol. 19517738, 2005.

\bibitem{Margulis-ppi74}
G.~A. Margulis, ``Probabilistic characteristics of graphs with large
  connectivity,'' \emph{Problems of Inform.\ Transm.}, vol.~10, no.~2, pp.
  101--108, 1974.

\bibitem{Russo-prf82}
L.~Russo, ``An approximate zero-one law,'' \emph{Prob. Th. and Related Fields},
  vol.~61, no.~1, pp. 129--139, 1982.

\bibitem{Talagrand-gafa93}
M.~Talagrand, ``Isoperimetry, logarithmic sobolev inequalities on the discrete
  cube, and margulis' graph connectivity theorem,'' \emph{Geometric \&
  Functional Analysis}, vol.~3, no.~3, pp. 295--314, 1993.

\bibitem{Friedgut-procams96}
E.~Friedgut and G.~Kalai, ``Every monotone graph property has a sharp
  threshold,'' \emph{Proc.\ Amer.\ Math.\ Soc.}, vol. 124, no.~10, pp.
  2993--3002, 1996.

\bibitem{Friedgut-jams99}
E.~Friedgut and J.~Bourgain, ``Sharp thresholds of graph properties, and the
  $k$-sat problem,'' \emph{J.\ Amer.\ Math.\ Soc.}, vol.~12, no.~4, pp.
  1017--1054, 1999.

\bibitem{Dinur-am05}
I.~Dinur and S.~Safra, ``On the hardness of approximating minimum vertex
  cover,'' \emph{Ann. of Math.}, pp. 439--485, 2005.

\bibitem{Zemor-94}
G.~Z{\'e}mor, ``Threshold effects in codes,'' in \emph{Algebraic Coding}.\hskip
  1em plus 0.5em minus 0.4em\relax Springer, 1994, pp. 278--286.

\bibitem{Tillich-cpc00}
J.-P. Tillich and G.~Z{\'e}mor, ``Discrete isoperimetric inequalities and the
  probability of a decoding error,'' \emph{Combinatorics, Probability and
  Computing}, vol.~9, no.~05, pp. 465--479, 2000.

\bibitem{Tillich-it04}
J.~Tillich and G.~Zemor, ``The {G}aussian isoperimetric inequality and decoding
  error probabilities for the {G}aussian channel,'' \emph{IEEE Trans.\ Inform.\
  Theory}, vol.~50, no.~2, pp. 328--331, Feb 2004.

\bibitem{Kumar-arxiv15v1}
S.~Kumar and H.~D. Pfister, ``Reed-{M}uller codes achieve capacity on erasure
  channels,'' 2015, [Online]. Available: http://arxiv.org/abs/1505.05123v1.

\bibitem{Kudekar-arxiv15v1}
S.~Kudekar, M.~Mondelli, E.~\c{S}a\c{s}o\u{g}lu, and R.~Urbanke,
  ``Reed-{M}uller codes achieve capacity on the binary erasure channel under
  {MAP} decoding,'' 2015, [Online]. Available:
  http://arxiv.org/abs/1505.05831v1.

\bibitem{Huffman-2003}
W.~C. Huffman and V.~Pless, \emph{Fundamentals of error-correcting
  codes}.\hskip 1em plus 0.5em minus 0.4em\relax Cambridge University Press,
  2003.

\bibitem{Berger-it96}
T.~Berger and P.~Charpin, ``The permutation group of affine-invariant extended
  cyclic codes,'' \emph{IEEE Trans.\ Inform.\ Theory}, vol.~42, no.~6, pp.
  2194--2209, Nov 1996.

\bibitem{Achlioptas-nature05}
D.~Achlioptas, A.~Naor, and Y.~Peres, ``Rigorous location of phase transitions
  in hard optimization problems,'' \emph{Nature}, vol. 435, no. 7043, pp.
  759--764, 2005.

\bibitem{Coja-Oghlan-stoc14}
A.~Coja-Oghlan, ``The asymptotic k-{SAT} threshold,'' in \emph{Proc.\ of the
  Annual ACM Symp.\ on Theory of Comp.}, ser. STOC '14.\hskip 1em plus 0.5em
  minus 0.4em\relax ACM, 2014, pp. 804--813.

\bibitem{Ding-arxiv14}
J.~Ding, A.~Sly, and N.~Sun, ``{P}roof of the satisfiability conjecture for
  large k,'' to appear in STOC 15, arXiv:1411.0650.

\bibitem{Talagrand-ap94}
M.~Talagrand, ``On {R}usso's approximate zero-one law,'' \emph{The Ann.\ of
  Prob.}, pp. 1576--1587, 1994.

\bibitem{Rossignol-ap06}
R.~Rossignol, ``Threshold for monotone symmetric properties through a
  logarithmic {S}obolev inequality,'' \emph{The Ann.\ of Prob.}, vol.~34,
  no.~5, pp. 1707--1725, 09 2006.

\bibitem{Ben-Or-rand90}
M.~Ben-Or and N.~Linial, ``Collective coin flipping,'' \emph{Randomness and
  Computation}, vol.~5, pp. 91--115, 1990.

\bibitem{Kahn-focs88}
J.~Kahn, G.~Kalai, and N.~Linial, ``The influence of variables on boolean
  functions,'' in \emph{Proc.\ IEEE Symp.\ on the Found.\ of Comp.\ Sci.}, Oct
  1988, pp. 68--80.

\bibitem{Bourgain-gafa97}
J.~Bourgain and G.~Kalai, ``Influences of variables and threshold intervals
  under group symmetries,'' \emph{Geometric \& Functional Analysis}, vol.~7,
  no.~3, pp. 438--461, 1997.

\bibitem{Kasami-it68*2}
T.~Kasami, S.~Lin, and W.~Peterson, ``Polynomial codes,'' \emph{IEEE Trans.\
  Inform.\ Theory}, vol.~14, no.~6, pp. 807--814, Nov 1968.

\bibitem{Delsarte-it70}
P.~Delsarte, ``On cyclic codes that are invariant under the general linear
  group,'' \emph{IEEE Trans.\ Inform.\ Theory}, vol.~16, no.~6, pp. 760--769,
  Nov 1970.

\bibitem{Berger-des99}
T.~Berger and P.~Charpin, ``The automorphism groups of {BCH} codes and of some
  affine-invariant codes over extension fields,'' \emph{Designs, Codes and
  Cryptography}, vol.~18, no. 1-3, pp. 29--53, 1999.

\bibitem{Ordentlich-it12}
O.~Ordentlich and U.~Erez, ``Cyclic-coded integer-forcing equalization,''
  \emph{IEEE Trans.\ Inform.\ Theory}, vol.~58, no.~9, pp. 5804--5815, 2012.

\bibitem{Measson-it09}
C.~M{\'e}asson, A.~Montanari, T.~J. Richardson, and R.~Urbanke, ``The
  generalized area theorem and some of its consequences,'' \emph{IEEE Trans.\
  Inform.\ Theory}, vol.~55, no.~11, pp. 4793--4821, Nov. 2009.

\end{thebibliography}

\end{document}